\documentclass[aos,preprint]{imsart}

\RequirePackage[OT1]{fontenc}
\usepackage{amsmath,amsfonts,amssymb,graphicx,color,amsthm,multirow,makecell}
\RequirePackage[numbers]{natbib}
\RequirePackage[colorlinks,citecolor=blue,urlcolor=blue]{hyperref}

\arxiv{arXiv:0000.0000}

\startlocaldefs
\newtheorem{theorem}{Theorem}

\newtheorem{corollary}{Corollary}

\theoremstyle{definition}

\endlocaldefs

\newcommand{\reals}{\mathbb{R}}

\begin{document}

\begin{frontmatter}
\title{Spillover Effects in Cluster Randomized Trials with Noncompliance}
\runtitle{Spillover Effects with Noncompliance}

\begin{aug}
\author{\fnms{Hyunseung} \snm{Kang}\ead[label=e1]{hyunseung@stat.wisc.edu}},
\author{\fnms{Luke} \snm{Keele}\ead[label=e2]{luke.keele@pennmedicine.upenn.edu}}

\runauthor{Kang and Keele}

\affiliation{University of Wisconsin-Madison and University of Pennsylvania}

\address{1220 Medical Sciences Center\\
Madison, WI 53706\\
\printead{e1}\\
}

\address{3400 Spruce St., Silverstein 4\\
Philadelphia, PA 19104\\
\printead{e2}\\
}
\end{aug}

\begin{abstract}
Cluster randomized trials (CRTs) are popular in public health and in the social sciences to evaluate a new treatment or policy where the new policy is randomly allocated to clusters of units rather than individual units. CRTs often feature both noncompliance, when individuals within a cluster are not exposed to the intervention, and individuals within a cluster may influence each other through treatment spillovers where those who comply with the new policy may affect the outcomes of those who do not. Here, we study the identification of causal effects in CRTs when both noncompliance and treatment spillovers are present. We prove that the standard analysis of CRT data with noncompliance using instrumental variables does not identify the usual complier average causal effect when treatment spillovers are present. We extend this result and show that no analysis of CRT data can unbiasedly estimate local network causal effects. Finally, we develop bounds for these causal effects under the assumption that the treatment is not harmful compared to the control. We demonstrate these results with an empirical study of a deworming intervention in Kenya. 
\end{abstract}



\end{frontmatter}

\section{Introduction}
Policy interventions are often evaluated by randomized controlled trials as random allocation of policy/treatment removes selection biases. However, there are two well-known complications in such trials. First, an individual's outcome may be influenced by him/her as well as his/her peers' treatment assignment, a phenomena known as interference \citep{cox_planning_1958}, and spillover effects \citep{Sobel2006-JASA} may occur. To mitigate concerns from spillovers, investigators often use clustered treatment assignments, usually in the form of cluster randomized trials (CRTs), to allow for arbitrary treatment spillovers within clusters \citep{cook_quasi_1979, murray_design_1998, imbens2008recent}. Second, subjects in the study may not comply with their randomized treatment assignment. For example, some may refuse to take the treatment or seek out treatment contrary to their treatment allocation. The method of instrumental variables (IV) is a well-understood framework to analyze randomized experiments with noncompliance; see \citet{angrist_instrumental_2001}, \citet{hernan_instruments_2006}, \citet{imbens_instrumental_2014} and \citet{baiocchi_instrumental_2014} for overviews. Increasingly, policy interventions exist at the intersection of these two complexities and our goal is to explore the consequences of spillovers and noncompliances in CRTs.

There is a large literature on both treatment spillovers and noncompliance, but typically these two topics are studied in isolation. In the literature on noncompliance in CRTs \citep{frangakis_clustered_2002,small_randomization_2008,jo_intention_2008,imai_essential_2009,schochet_estimation_2011}, treatment spillovers are generally considered a nuisance and their effects are minimized by clustered treatment assignment. For example, \citet{small_randomization_2008} notes that ``the issue of interference...does not arise'' in CRTs (Section 2.1 of \citet{small_randomization_2008}) and \citet{imai_essential_2009} assumes no interference (Assumption 3 in Section 6.2 of \citet{imai_essential_2009}). In the growing literature on interference and spillover effects \citep{halloran_study_1991,halloran_causal_1995,Sobel2006-JASA,hong2006evaluating,Rosenbaum2007-JASA,hudgens_toward_2008,TchetgenTchetgenVanderWeele2010,VanderWeele2008-StatMed,Vanderweele:2013,BowersFredricksonPanagopoulos2013,AronowSamii2017-AOAS}, the primarily focus is on defining or estimating network causal quantities and treatment compliance is ignored. Some notable exceptions include \citet{Sobel2006-JASA}, \citet{hong2015causality},\citet{forastiere2016identification}, \citet{kang_peer_2016}, and \citet{imai2018causal}. \citet{Sobel2006-JASA} highlighted problems when noncompliance and interference are both present in non-clustered designs, but did not provide methods to estimate the spillover effects formalized in 
\citet{hudgens_toward_2008}. \citet{forastiere2016identification} studied noncompliance and interference under a Bayesian paradigm. \citet{kang_peer_2016} and \citet{imai2018causal} studied noncompliance and interference under multi-level designs where randomization occurs multiple times and in different hierarchical levels. In contrast, CRTs have only have one level of randomization, which is at the cluster level. \citet{hong2015causality}, in Chapter 15.2, studied noncompliance and interference under the mediation framework.

Here, we study the identification of causal effects in CRTs when both noncompliance and interference are present. First, we show that the standard analysis of CRTs with noncompliance using the instrumental variables (IV) framework, following the methods in \citet{small_randomization_2008}, \citet{jo_intention_2008}, and \citet{schochet_estimation_2011}, does not identify the usual causal estimand known as the complier average causal effect (CACE) in the presence of spillover effects. Second, we extend our result and show that no analysis of CRTs can unbiasedly estimate network effects under noncompliance when interference is present; specifically, we show that there does not exist a function of the observed data from CRTs that can unbiasedly estimate the causal parameters. The second result suggests that CRTs, as an experimental design, are generally unsuited to learn about treatment spillovers in the presence of noncompliance and a new experimental design or stronger assumptions about the data are necessary to study network effects in the presence of noncompliance. Third, we show that for data from CRTs with noncompliance, investigators can estimate bounds on spillover and total effects under an assumption about treatment monotonicity.

\section{Preliminaries: Notation and Assumptions} 

We structure the notation and motivate the assumptions using a public health intervention called the Primary School Deworming Project (PSDP), which we analyze in later sections.
PSDP was conducted by a Dutch nonprofit organization, International Christelijk Steunfonds Africa (ICS), in cooperation with the Busia District Ministry of Health office \citep{miguel2004worms}. The intervention consisted of deworming treatments for intestinal helminths such as hookworm, roundworm, whipworm, and schistosomiasis, delivered to school children in southern Busia,  an area in Kenya with the highest helminth infection rates. The primary goal was to evaluate whether the deworming treatments on reduced intestinal helminths infections among students. The intervention was evaluated using a CRT, where the new treatment was randomly allocated to different schools that served as clusters, and as such, all students in treated schools were offered deworming treatments in the form of oral medications. The medication was believed to have both direct and spillover effects; it not only killed helminths among those who took them (i.e. direct effect), but also decreased disease transmission among peers by reducing the number of helminths in the environment (i.e. a spillover effect) since students were primarily exposed to intestinal helminths through environmental exposures such as outdoor defecation and contact with infected fresh water. Also, unit level noncompliance occurred as the investigators were required to obtain parental consent for the study. Even when parental consent was obtained, students in treated schools did not always take the deworming treatments.

\subsection{Notation} 
\label{sec:notation}

There are $J$ clusters indexed by $j=1,\ldots,J$ and for each cluster $j$, there are $n_{j}$ individuals, indexed by $i=1,\ldots,n_{j}$. There are $N = \sum_{j=1}^{J} n_j$ total individuals in the study population. Let $Z_{j} \in \{0,1\}$ denote the treatment assignment of cluster $j$ where $Z_{j} = 1$ indicates that cluster $j$ was assigned to treatment and $Z_{j} = 0$ indicates that cluster $j$ was assigned to control. Let $D_{ji} \in \{0,1\}$ denote the observed treatment receipt of individual $i$ in cluster $j$ where $D_{ji} = 1$ indicates that individual $i$ actually took the treatment and $D_{ji} = 0$ indicates that individual $i$ did not take the treatment (i.e. control). Note that while the treatment is assigned at the cluster level, the decision to comply occurs at the unit level and hence, there is an extra subscript $i$ in $D_{ji}$. This is important because in the aforementioned PDSP study, schools were assigned to the intervention, but each student and parent in a school could choose to comply with the intervention. Let $Y_{ji} \in \reals$ represent the observed outcome of individual $i$ in cluster $j$. Let $\mathbf{Y} = (Y_{11},Y_{12},\ldots, Y_{Jn_J})$, $\mathbf{D} = (D_{11}, D_{12}, \ldots, D_{Jn_J})$, and $\mathbf{Z} = (Z_1,\ldots, Z_J)$ be the outcome, compliance, and treatment assignment vectors, respectively. 

Let $\mathcal{B}^p = \{0,1\}^p$ be the set of all binary vectors of length $p$. For any vector $\mathbf{v} \in \mathcal{B}^p$ and integer $k \in \{1,\ldots,p\}$, let $\mathbf{v}_{-k}$ denote the vector $\mathbf{v}$ with the $k$th index removed. Let $I(\cdot)$ denote an indicator function where $I(\cdot) = 1$ if the event inside the indicator function is true and $0$ otherwise. Finally, let $\mathbf{1} = (1,\ldots,1) \in \mathcal{B}^p$ be a vector of ones.

\subsection{Potential Outcomes, SUTVA, and Treatment Compliance Heterogeneity} 
\label{sec:A_int}

We define causal effects using the potential outcomes notation \citep{neyman_application_1990,rubin_estimating_1974}. For each $z_{j} \in \{0,1\}$, let $D_{ji}^{(z_{j})}$ denote individual $i$'s potential treatment receipt if his cluster $j$ were assigned treatment $z_j$. Let $\mathbf{D}_{j}^{(z_{j})} = (D_{j1}^{(z_{j})},\ldots,D_{jn_{j}}^{(z_{j})})$ denote the vector of potential treatment receipts for cluster $j$. For each $z_{j} \in \{0,1\}$, $d_{ji} \in \{0,1\}$ and $\mathbf{d}_{j-i} = (d_{j1},\ldots, d_{ji-1}, d_{ji+1}, \ldots,d_{jn_{j}}) \in \mathcal{B}^{n_{j} - 1}$, let $Y_{ji}^{(z_{j}, d_{ji}, \mathbf{d}_{j-i})}$ denote individual $i$'s potential outcome if his cluster $j$ were assigned treatment $z_{j}$, his treatment compliance were $d_{ji}$, and his peers' treatment receipt were $\mathbf{d}_{j-i}$. Also, let $Y_{ji}^{(z_{j}, D_{ji}^{(z_{j})}, \mathbf{D}_{j-i}^{(z_{j})})}$ denote individual $i$'s potential outcome if his cluster $j$ were assigned treatment $z_{j}$ and he and his peers' treatment receipts were their ``natural'' compliances $D_{ji}^{(z_{j})}$ and $\mathbf{D}_{j-i}^{(z_{j})}$, respectively. Let $\mathbf{Y}_{j}^{(z_{j})} = (Y_{j1}^{(z_{j}, D_{j1}^{(z_{j})}, \mathbf{D}_{j-1}^{(z_{j})})},\ldots,Y_{ji}^{(z_{j}, D_{jn_i}^{(z_{j})}, \mathbf{D}_{j-n_i}^{(z_{j})})})$ denote the vector of potential outcomes for cluster $j$. Let $\mathcal{F} = \{ Y_{ji}^{(z_{j},d_{ji}, \mathbf{d}_{j-i})},D_{ji}^{(z_{j})} \mid z_{j}, d_{ji} \in \{0,1\}, \mathbf{d}_{j-i} \in \mathcal{B}^{n_{j} - 1},j=1,\ldots,J, i =1,\ldots,n_{j}\}$ be the set of all potential outcomes, which we assume to be fixed and unknown. 

Our notation assumes what \citet{Sobel2006-JASA} called partial interference, which is a particular violation of the stable unit treatment value assumption (SUTVA) \citep{rubin_estimating_1974}. Specifically, the potential treatment receipt and outcome, $D_{ji}^{(z_{j})}$ and $Y_{ji}^{(z_{j},d_{ji},\mathbf{d}_{j-i})}$ respectively, are only affected by values from cluster $j$, specifically $z_{j}$, $d_{ji}$, and $\mathbf{d}_{j-i}$; the potential treatment receipt and outcome are not affected by values from other clusters $j' \neq j$, say by $z_{j'}$, $d_{j'i}$, and $\mathbf{d}_{j'-i}$. Also, in a CRT, clustered treatment assignment implies that we only observe the potential treatment receipt of individual $i$ in cluster $j$ when all individuals in cluster $j$ are assigned treatment, $D_{ji}^{(1)}$, or when all individuals are assigned control $D_{ji}^{(0)}$; we do not observe individual $i$'s potential treatment receipt when some individuals in cluster $j$ are assigned treatment and rest are assigned control. For example, we do not observe $D_{ji}^{(0,1,\ldots,1)}$ where everyone except individual $i$ in cluster $j$ are assigned treatment. But, due to unit level noncompliance where only some individuals actually end up taking the treatment, our setting allows a fraction of individuals within a cluster to take the treatment so that we have variation in treatment receipt within a cluster. This is in contrast to the two-stage randomization design of \citet{hudgens_toward_2008} which has two stages of randomization, both at the cluster and the unit level, and treatment receipt is assumed to be randomly assigned.  We consider a context where treatment receipt is self-selected because of noncompliance.

Moreover, under interference, treatment compliance may be heterogeneous. The extant literature on CRTs with noncompliance has used potential outcomes of the form $Y_{ji}^{(z_{j}, d_{ji})}$ where the outcome of individual $i$ is a function of his own treatment receipt $d_{ji}$ and his cluster treatment assignment $z_j$ \citep{frangakis_clustered_2002,small_randomization_2008,jo_intention_2008,imai_essential_2009,schochet_estimation_2011}. The potential outcomes notation in this paper, $Y_{ji}^{(z_{j}, d_{ji}, \mathbf{d}_{j-i})}$, is a generalization of prior literature's notation because individual $i$'s potential outcome is affected both by his own compliance $d_{ji}$ and the compliances of his peers $\mathbf{d}_{j-i}$. As a concrete example, in a hypothetical cluster of size $n_j=2$ that is assigned treatment $z_j = 1$, individual $i=1$'s potential outcome could be $Y_{j1}^{(1,D_{j1}^{(1)},D_{j2}^{(1)})} = Y_{j1}^{(1,1,0)}$ if individual $i=1$ actually took the treatment so that $D_{j1}^{(1)} =1 $ and individual $i=2$ did not take the treatment so that $D_{j2}^{(1)} = 0$. Alternatively, individual $i=1$'s potential outcome could be $Y_{j1}^{(1,0,0)}$ if individual $i=1$ refused treatment so that $D_{j1}^{(1)} = 0$. Prior literature's notation would assume these two potential outcomes are equal, $Y_{j1}^{(z_j,0,1)} = Y_{j1}^{(z_j,0,0)}$. In contrast, our notation makes it explicitly clear that interference can exist and if it exists, it originates from heterogeneous noncompliance of each individual unit.

\subsection{Assumptions}
Next, we review the assumptions that are standard in the literature on noncompliance and interference; detailed discussions can be found in \citet{angrist_identification_1996} and \citet{baiocchi_instrumental_2014} (for noncompliance) and \citet{hudgens_toward_2008} (for interference). We invoke this set of assumptions to understand what can be identified from CRTs under extant assumptions.
\begin{itemize}
\item[(A1)] \emph{Cluster Randomized Assignment}. Let $\mathcal{Z}$ be the set that consists of vectors $\mathbf{z} = (z_{1},\ldots,z_{J}) \in \mathcal{B}^J$ such that $m = \sum_{j=1}^{J} z_{j}$ for $0< m < J$. Then, $m$ out of $J$ clusters are randomly assigned treatment and $J -m$ are assigned control. 
\[
P(\mathbf{Z} = \mathbf{z} \mid \mathcal{F}, \mathcal{Z}) = P(\mathbf{Z} = \mathbf{z} \mid \mathcal{Z}) = \frac{1}{{J \choose m}} 
\]
\item[(A2)] \emph{Non-Zero Causal Effect of $Z$ on $D$}. The treatment assignment has a non-zero average causal effect on compliance.
\[
\tau_D = \frac{1}{N} \sum_{j=1}^{J} \sum_{i=1}^{n_{j}} D_{ji}^{(1)} - D_{ji}^{(0)} \neq 0 
\]
\item[(A3)] \emph{Network Exclusion Restriction}. Given individual $i$'s compliance $d_{ji}$ and the compliances of others $\mathbf{d}_{j-i} \in \mathcal{B}^{n_{j} - 1}$, the treatment assignment $z_{j}$ has no impact on the potential outcome of individual $i$.
\[
Y_{ji}^{(1,d_{ji}, \mathbf{d}_{j-i})} = Y_{ji}^{(0,d_{ji}, \mathbf{d}_{j-i})} \equiv Y_{ji}^{(d_{ji}, \mathbf{d}_{j-i})}
\]
\item[(A4)] \emph{Monotonicity}. For every individual $i$ in cluster $j$, we have $D_{ji}^{(0)} \leq D_{ji}^{(1)}$.
\item[(A5)] \emph{Stratified Interference}. Given individual $i$'s compliance $d_{ji}$, his cluster treatment assignment $z_j$, and a whole number $k \in \{0,1,\ldots,n_{j} - 1\}$, the potential outcomes of individual $i$ are equivalent when exactly $k$ of his peers are taking the treatment.
\[
Y_{ji}^{(z_{j}, d_{ji}, \mathbf{d}_{j-i})} = Y_{ji}^{(z_{j}, d_{ji}, \widetilde{\mathbf{d}}_{j-i})} \equiv Y_{ji}^{(z_j, d_{ji}, [k])}, \quad{} \forall \mathbf{d}_{j-i}, \widetilde{\mathbf{d}}_{j-i} \in \mathcal{B}^{n_{j} - 1} \text{ where } k = \sum_{i' \neq i} d_{ji'} = \sum_{i' \neq i} \widetilde{d}_{ji'}
\]
\end{itemize}
We briefly comment on these assumption within the context of the PSDP, our motivating example. Assumption (A1) is approximately satisfied by the design of the PSDP where the deworming intervention was randomized to schools (i.e. clusters); see \citet{miguel2004worms} and \citet{hicks2015commentary} for additional details on the treatment assignment process. Assumption (A1) also allows us to test assumption (A2) in the PSDP by taking the stratified difference-in-means of the observed compliance values $D_{ji}$ between the treated and control clusters; see Section \ref{sec:app} for a numerical illustration. Assumption (A3), like the usual exclusion restriction in the noncompliance literature \citep{angrist_identification_1996}, cannot be tested with data because it requires observing potential outcomes under both treatment $z_{j} = 1$ and control $z_{j} = 0$; typically, subject-matter expertise must be used to justify this assumption. In the PSDP, assumption (A3) implies that the random assignment of the deworming treatments had an effect on the outcome, say a student's infection status, only through students taking the oral medications. If the random assignment induced better hygiene at the school leading to lower infection rates and this is not captured by the treatment receipt values $d_{ji}$ or $\mathbf{d}_{j-i}$, the exclusion restriction would be violated. 

Assumption (A4) can be interpreted by partitioning the study population into four groups, compliers (CO), always-takers (AT), never-takers (NT), and defiers (DF) \citep{angrist_identification_1996}. In the PSDP, compliers are students who follow the  intervention, $D_{ji}^{(1)} = 1$ and $D_{ji}^{(0)} = 0$. Always-takers are students who always take the deworming treatments, irrespective of the intervention, $D_{ji}^{(1)} = D_{ji}^{(0)} = 1$. Never-takers are students who never take the deworming treatments, irrespective of the intervention, $D_{ji}^{(1)} = D_{ji}^{(0)} = 0$. Defiers are the opposite of compliers in that they systematically defy the intervention, $D_{ji}^{(1)} = 0$ and $D_{ji}^{(0)} = 1$. Assumption (A4) implies that there are no defiers in the PSDP population.

In many CRTs, including the PSDP, assumption (A4) can be satisfied by the study design by denying individuals in clusters that are randomized to the control to seek out the treatment; in the PSDP, the students in the control clusters did not have access to the new deworming medications. This is commonly referred to as one-sided noncompliance and is formalized as assumption (A4.1).
\begin{itemize}
\item[(A4.1)] \emph{One-Sided Noncompliance}. For every individual $i$ in cluster $j$, we have $D_{ji}^{(0)} = 0$
\end{itemize}
An implication of assumption (A4.1) is that there are no always-takers and defiers in the study population.

Finally, assumption (A5) states that conditional on individual $i$'s compliance to treatment $d_{ji}$ and his cluster treatment assignment $z_{j}$, individual $i$'s potential outcome $Y_{ji}^{(z_{j}, d_{ji}, \mathbf{d}_{j-i})}$ only depends on the number of individual $i$'s peers who took the treatment, not necessarily who took the treatment. \citet{hudgens_toward_2008} talks about the plausibility of this assumption in practice, especially in infectious disease studies like the PSDP, as well as the statistical importance of having this assumption to estimate standard errors and conduct asymptotics \citep{liu_large_2014}. In our work, we make this assumption not out of technical necessity, but out of better interpretability; see Section \ref{sec:ident} for details. Also, note that combining assumptions (A3) and (A5) leads to the original stratified interference assumption stated in \citet{hudgens_toward_2008}.

\subsection{Review: Causal Effects under Noncompliance and Interference} 
\label{sec:review_effect}
Next, we review two causal estimands of interest when noncompliance or interference are present. First, under noncompliance, but without interference, the potential outcome $Y_{ji}^{(z_j, d_{ji}, \mathbf{d}_{j-i})}$ collapses to $Y_{ji}^{(z_j, d_{ji})}$ and two causal effects of primary interest are the average intent-to-treat (ITT) effects and the ratio of these effects to identify the complier average treatment effect (CACE) \citep{imbens_identification_1994,angrist_identification_1996}. The first ITT effect is with respect to $Y$ and is denoted as $\tau_Y$:
\[
\tau_{Y} = \frac{1}{N} \sum_{j=1}^{J} \sum_{i=1}^{n_{j}} Y_{ji}^{(1,D_{ji}^{(1)})} - Y_{ji}^{(0,D_{ji}^{(0)})}
\]
$\tau_Y$ is the population average effect of treatment assignment, not treatment receipt, on the outcome. The second ITT effect is with respect to $D$ and is denoted as $\tau_D$:
\[
\tau_{D} = \frac{1}{N} \sum_{j=1}^{J} \sum_{i=1}^{n_{j}} D_{ji}^{(Z_{ji}^{(1)})} - D_{ji}^{(Z_{ji}^{(0)})}
\]
$\tau_D$ is also known as the compliance rate. Let $\tau$ be the ratio of the average ITT effect $\tau_Y$ to the average effect of treatment assignment on compliance $\tau_D$, 
\[
\tau = \frac{\tau_Y}{ \tau_D}
\]
where we assume (A2) holds so that $\tau$ is well-defined. Then, under (A2)-(A4) without interference, $\tau$ equals the complier average treatment effect.
\begin{equation}
\label{eq:cace}
\tau = \frac{1}{N_{\rm CO}} \sum_{j=1}^J \sum_{i=1}^{n_j} (Y_{ji}^{(1)} - Y_{ji}^{(0)}) I(\text{$ji$ is CO}) 
\end{equation}
In equation \eqref{eq:cace}, $N_{\rm CO}$ is the total number of compliers in the study population, which is non-zero because of (A2) and (A4). The CACE is also referred to as a local effect because it only represents the treatment effect of a particular subgroup in the population \citep{imbens_identification_1994}. The CACE can be estimated by plugging in estimates of $\tau_Y$ and $\tau_D$ that make up  $\tau$. Specifically, $\tau_Y$ and $\tau_D$ can each be unbiasedly estimated under (A1) by taking the stratified difference-in-means of the observed outcomes $Y_{ji}$ and treatment receipts $D_{ji}$,  respectively, between the treated and control clusters; see Section \ref{sec:est} for details. This ratio estimator is known as the ``Wald'' or the IV estimator, and it is a special case of the two stage least squares (TSLS) estimator in the literature. However, when interference is present, it is unclear whether the ratio can be interpreted as the CACE. We explore this in Section \ref{sec:ident}.

Second, under interference, but with full compliance, the potential outcome $Y_{ji}^{(z_j, d_{ji}, \mathbf{d}_{j-i})}$ collapses to $Y_{ji}^{(d_{ji}, \mathbf{d}_{j-i})}$ and causal effects of interest are the total, direct, and spillover effects, denoted as ${\rm TE}_{ji}$, ${\rm DE}_{ji}$, and ${\rm PE}_{ji}$, respectively \citep{hudgens_toward_2008}. Formally, under stratified interference (A5), given two vectors of treatment receipts $\mathbf{d}_{j-i}, \mathbf{d}_{j-i}' \in \mathcal{B}^{n_{j} - 1}$ and two whole numbers $k_1, k_0 \in \{0,1,\ldots,n_{j}-1\}$ that count the number of $i$'s peers that took the treatment, i.e. $k_1 = \sum_{i' \neq i} d_{ji'}$ and $k_0 = \sum_{i' \neq i} d_{ji'}'$, the individual total, direct, and spillover/peer effects are defined as
\begin{align*}
{\rm TE}_{ji}(1,k_1;0,k_0) = Y_{ji}^{(1,\mathbf{d}_{j-i})} - Y_{ji}^{(0,\mathbf{d}_{j-i}')} \\
{\rm DE}_{ji}(1,k_1; 0,k_1) = Y_{ji}^{(1,\mathbf{d}_{j-i})} - Y_{ji}^{(0,\mathbf{d}_{j-i})} \\
{\rm PE}_{ji}(0,k_1;0,k_0) = Y_{ji}^{(0,\mathbf{d}_{j-i})} - Y_{ji}^{(0,\mathbf{d}_{j-i}')} \\
{\rm PE}_{ji}(1,k_1;1,k_0) = Y_{ji}^{(1,\mathbf{d}_{j-i})} - Y_{ji}^{(1,\mathbf{d}_{j-i}')}
\end{align*}
${\rm TE}_{ji}(1,k_1;0,k_0)$ is individual $i$'s total casual effect if individual $i$ and $k_1$ of his peers took the treatment versus if individual $i$ did not take the treatment while $k_0$ of his/her peers did take the treatment. ${\rm DE}_{ji}(1,k_1;0,k_1)$ is individual $i$'s direct causal effect if individual $i$ took the treatment versus if individual $i$ did not take the treatment and his peers' treatment receipt was fixed at $k_1$, i.e. if $k_1$ of his peers took the treatment. ${\rm PE}_{ji}(0,k_1;0,k_0)$ is individual $i$'s spillover causal effect if $k_1$ of his peers took the treatment versus if $k_0$ of his peers took the treatment and individual $i$'s treatment receipt remained fixed at $d_{ji} = 0$, i.e. individual $i$ did not take the treatment. ${\rm PE}_{ji}(1,k_1;1,k_0)$ is the same as ${\rm PE}_{ji}(0,k_1;0,k_0)$ except individual $i$'s treatment receipt is fixed at $d_{ji} = 1$, i.e. individual $i$ took the treatment. Without noncompliance, the population averages of individual total, direct, and spillover effects can be identified under a two-stage randomization design of \citet{hudgens_toward_2008} where some clusters are randomly allocated to the ``$k_1$ policy'' that assigns $k_1$ individuals in the cluster to treatment while the rest are assigned the ``$k_0$ policy'' where $k_0$ individuals in the cluster are treated and each individual are assigned treatment based on his cluster randomization policy.


\section{Identification}
\subsection{Target Causal Estimands} 
\label{sec:target}
To discuss the results in the paper, we introduce three causal estimands which may be of natural interest in CRTs when both noncompliance and interference are present and to the best of our knowledge, have not been discussed in either literature on noncompliance or interference. Broadly speaking, our three proposed causal estimands are the network estimands in \citet{hudgens_toward_2008} specific to subgroups of compliers, always-takers, and never-takers in \citet{angrist_identification_1996}. We also note that while a richer combinations of these two literatures and their effects are possible, we only discuss three combinations as they are the most relevant in our paper. 

Let $n_j^{\rm AT}$, $n_{j}^{\rm CO}$, and $n_{j}^{NT}$ be the number of always-takers, compliers, and never-takers, respectively, in cluster $j$ and let $N^{\rm AT} = \sum_{j=1}^{J} n_j^{\rm AT}$,  $N^{\rm CO} = \sum_{j=1}^{J} n_j^{\rm CO}$, and $N^{\rm NT} = \sum_{j=1}^J n_j^{\rm NT}$ be the total number of always-takers, compliers, and never-takers, respectively, in the population. For each cluster $j$, let $k_{0j}, k_{1j} \in \{0,1,\ldots,n_j -1\}$, $ k_{0j} < k_{1j}$, indicate two different numbers of peers actually taking the treatment, similar to $k_0$ and $k_1$ in Section \ref{sec:review_effect}. For instance, in a hypothetical cluster of size $n_j = 10$, $k_{0j} = 5$ would indicate that $5$ peers are taking the treatment in cluster $j$ and $k_{1j} = 7$ would indicate that $7$ peers are taking the treatment in cluster $j$. Let $\mathbf{k}_0 = (k_{01},\ldots,k_{0J})$ and $\mathbf{k}_1 = (k_{11},\ldots,k_{1J})$ be collections of $k_{0j}$ and $k_{1j}$, respectively, across all $J$ clusters. We define the population average total effect of treatment on the outcome among compliers (CO)
\begin{equation} \label{eq:avg_TE}
\overline{\rm TE}^{\rm CO}(1,\mathbf{k}_1;0,\mathbf{k}_0) = \frac{1}{N^{\rm CO}} \sum_{j=1}^{J} \sum_{i=1}^{n_{j}} {\rm TE}_{ji}(1,k_{1j};0,k_{0j}) I(\text{ji is CO})
\end{equation}
Each individual total effect in equation \eqref{eq:avg_TE}, i.e. ${\rm TE}_{ji}(1,k_{1j};0,k_{0j}) I(\text{ji is CO})$, is the total effect of complier $i$ when $k_{1j}$ of his peers are taking the treatment versus him not taking the treatment while $k_{0j}$ of his peers are taking the treatment. Note that each cluster has different numbers of peers taking the treatment; this is in contrast with the effects defined by \citet{hudgens_toward_2008} where because the treatment policies are fixed, there are (two) fixed number of people taking the treatment. Also, while individual $i$ is a complier, his peers may be a mixture of compliers, always-takers, and never-takers. In the PSDP, we expect that $\overline{\rm TE}^{\rm CO}(1,\mathbf{k}_1;0,\mathbf{k}_0)  > 0$, that is the deworming treatment has a net positive benefit to complier students who take the assigned medication.

Also, we define the population average spillover effect of treatment on the outcome among always-takers (AT)
\begin{equation} \label{eq:avg_PE1}
\overline{\rm PE}^{\rm AT}(1,\mathbf{k}_1; 1,\mathbf{k}_0) = \frac{1}{N^{\rm AT}} \sum_{j=1}^{J} \sum_{i=1}^{n_{j}} {\rm PE}_{ji}(1,k_{1j};1,k_{0j}) I(\text{ji is AT})
\end{equation}
Equation \eqref{eq:avg_PE1} implicitly assumes that there is at least one always-taker in the population so that $N^{\rm AT} \neq 0$. Each individual spillover effect in equation \eqref{eq:avg_PE1}, i.e. ${\rm PE}_{ji}(1,k_{1j};1,k_{0j}) I(\text{ji is AT})$, is the spillover effect of always-taker $i$ when when $k_{1j}$ versus $k_{0j}$ of always-taker $i$'s peers in cluster $j$ take the treatment while he takes the treatment. Note that similar to equation \eqref{eq:avg_TE}, always-takers' peers may be a mixture of compliers, never-takers, and always-takers. If $\overline{\rm PE}^{\rm AT}(1,\mathbf{k}_{1};1,\mathbf{k}_0) =0$, having additional always-takers' peers take the treatment does not affect, on average, the always-takers' outcomes. When $\overline{\rm PE}^{\rm AT}(1,\mathbf{k}_1;1,\mathbf{k}_0) > 0$, having additional always-takers' peers take the treatment benefits the always-takers' outcomes. In the PSDP, we expect $\overline{\rm PE}^{\rm AT}(1,\mathbf{k}_1;1,\mathbf{k}_0) \geq 0$, that is more always-takers' peers taking the deworming medication is not harmful to the always-taker students.

Finally, we define a parallel effect to $\overline{\rm PE}^{\rm AT}(1,\mathbf{k}_1;1,\mathbf{k}_0)$, the population average spillover effect among never-takers.
\begin{equation} \label{eq:avg_PE0}
\overline{\rm PE}^{\rm NT}(0,\mathbf{k}_1;0,\mathbf{k}_0) =  \frac{1}{N^{\rm NT}} \sum_{j=1}^{J} \sum_{i=1}^{n_j} {\rm PE}_{ji}(0,k_{1j};0,k_{0j}) I(\text{ji is NT})
\end{equation}
Equation \eqref{eq:avg_PE0} implicitly assumes that there is at least one never-taker in the population so that $N^{\rm NT} \neq 0$. Each individual spillover effect in equation \eqref{eq:avg_PE0}, i.e. ${\rm PE}_{ji}(0,k_{1j};0,k_{0j}) I(\text{ji is NT})$, is the spillover effect of never-taker $i$ when $k_{1j}$ versus $k_{0j}$ of never-taker $i$'s peers in cluster $j$ take the treatment while never-taker $i$ does not take the treatment. Like the spillover effect among always-takers in equation \eqref{eq:avg_PE1}, $\overline{\rm PE}^{\rm NT}(0,\mathbf{k}_1;0,\mathbf{k}_0)$ can be thought of as the additional effect on never-takers when more of their peers take the treatment. In the PSDP, the spillover effect among never-takers is likely non-negative, i.e. $\overline{\rm PE}^{\rm NT}(0,\mathbf{k}_1;0,\mathbf{k}_0) \geq 0$, since peers of never-takers who take the deworming treatment reduce the likelihood of infections to everyone, including the never-takers who refuse to take the medication.

The spillover effect among never-takers may be a practically useful estimand because it helps investigators understand how the treatment spills over to individuals who will never take the treatment irrespective of the intervention assignment. For example, in the PSDP, the never-takers students could be students who cannot take the new medication due to side effects or students who are immunocompromised. In vaccine studies, understanding how the peers of never-takers taking the treatment affect the never-takers' outcomes could be useful for informing vaccination policies, say by understanding the effect of herd immunity of vaccines among individuals who refuse (or medically cannot) to take the vaccine. This is in contrast to the usual local estimands in noncompliance, such as the CACE where its significance and merit are debated \citep{Deaton:2010,Imbens:2010,imbens_instrumental_2014,baiocchi_instrumental_2014,Swanson:2014}.

\subsection{Standard IV Analysis Does Not Identify CACE} 
\label{sec:ident}
The first identification result we prove is that the standard IV estimator based on a ratio of the ITT effects as described in Section \ref{sec:review_effect} does not identify the CACE when interference is present. 
\begin{theorem} \label{thm:gen} Suppose assumptions (A1)-(A5) hold and suppose there is at least one complier in cluster $j$. Then, the ratio of $\tau_Y$ to $\tau_D$, $\tau$, is a mixture of causal effects among compliers, always-takers and never-takers.\begin{align} \label{eq:tau_int}
\tau &= \overline{\rm TE}^{\rm CO}(1,\mathbf{k}_{1} - \mathbf{1}; 0,\mathbf{k}_{0})  + \frac{N^{\rm AT}}{N^{\rm CO}}\cdot \overline{\rm PE}^{\rm AT}(1,\mathbf{k}_{1} - \mathbf{1}; 1,\mathbf{k}_{0} - \mathbf{1})+ \frac{N^{\rm NT}}{N^{\rm CO}} \cdot \overline{\rm PE}^{\rm NT} (0,\mathbf{k}_{1}; 0,\mathbf{k}_{0}) \end{align}
where $\mathbf{k}_1 = (n_{1}^{\rm AT} + n_{1}^{\rm CO}, \ldots, n_{J}^{\rm AT} + n_{J}^{\rm CO})$ and $\mathbf{k}_{0} = (n_{1}^{\rm AT},\ldots, n_{J}^{\rm AT})$
\end{theorem}
Theorem \ref{thm:gen} highlights that when both noncompliance and interference are present in a CRT,  the standard Wald estimator must be interpreted as a mixture of treatment effects from different subgroups. Each effect in the mixture is associated with a particular subgroup, such as the spillover effects being associated with always-takers or never-takers instead of compliers. This is because, by definition, spillover effects fix individuals' own treatment assignments, but vary the peers' treatment assignments. Hence, to observe spillover effects under noncompliance, individuals must always take the treatment (or control) irrespective of the intervention assignment and only always-takers and never-takers have this trait. In contrast, the definition of total effects vary individuals and their peers' treatment assignments. To observe total effects under noncompliance, individuals must comply with their treatment assignments so that there is variation in their treatment receipts and only compliers have this trait.

Intuitively, the result in Theorem \ref{thm:gen} is based on the fact that under interference, the noncompliers' potential outcomes, i.e. the never-takers and the always-takers' outcomes, are affected by the behavior of their complier peers. For example, for never-takers in cluster $j$, if their cluster $j$ is assigned treatment, there are $n_{j}^{\rm CO} + n_j^{\rm AT}$ peers of never-takers taking the treatment and if their cluster $j$ is assigned control, there are $n_{j}^{\rm AT}$ peers of them taking the treatment. The never-takers potential outcomes under these two randomization arms are different because more peers are taking the treatment when the cluster is randomized to treatment than control. Consequently, the difference in their potential outcomes between treatment and control is non-zero and remain as spillover effects $\overline{\rm PE}^{\rm NT} (0,\mathbf{k}_{1}; 0,\mathbf{k}_{0})$ in the expression for $\tau$. In contrast, without interference, the never-takers' potential outcomes do not depend on their peers' treatment receipt. Therefore, their outcomes when the cluster is assigned treatment or control is always the outcome when they themselves do not take the treatment. This means that the difference in the two potential outcomes will be zero and the never-takers' treatment effects do not appear in $\tau$ of Theorem \ref{thm:gen}. A formal argument is in the Appendix, where we also prove Theorem \ref{thm:gen} with and without stratified interference (A5).

We also take a moment to relate our results to those in the literature on interference and noncompliance. First, a key difference between the population effects in Theorem \ref{thm:gen} and those in \citet{hudgens_toward_2008} is that \citet{hudgens_toward_2008} studied contrasts between two fixed treatment policies applied to clusters, say treatment policy $1$ that randomly assigns 50\% of individuals in a cluster to treatment, and treatment policy $2$ that randomly assigns 30\% of individuals in a cluster to treatment. Additionally, their population average effects averaged over all individuals in the study. In our setting, because of noncompliance, each cluster may have different numbers of people who (non-randomly) take the treatment and thus, the number of individuals actually taking the treatment varies across clusters. Also, our population average effects  $\overline{\rm TE}^{\rm CO}(1,\mathbf{k}_{1} - \mathbf{1}; 0,\mathbf{k}_{0})$, $\overline{\rm PE}^{\rm AT} (0,\mathbf{k}_{1}; 0,\mathbf{k}_{0})$, and $\overline{\rm PE}^{\rm NT} (0,\mathbf{k}_{1}; 0,\mathbf{k}_{0})$ are averaged over subgroups in the population. Our population average effects become those in \citet{hudgens_toward_2008} if everyone in the study is a complier; under this case, our population average effects would become contrasts between two treatment policies, one where everyone in a cluster is treated and one where everyone is not treated, and $\tau$ reduces to the population average total effect of \citet{hudgens_toward_2008}. Second, Theorem \ref{thm:gen} can be seen as a generalization of the classical identification results of CACE \citep{imbens_identification_1994,angrist_identification_1996} that allows for interference. In particular, without interference, the spillover effects in $\tau$ equal zero, the total effect is the direct effect \citep{hudgens_toward_2008} and $\tau$ in Theorem \ref{thm:gen} reduces to the CACE.

Finally, the peers' ``observable''  values of $Y$ in a CRT with interference do not involve the never-takers $n_{j}^{\rm NT}$ when stratified interference (A5) holds. In fact, the number of compliers $n_j^{\rm CO}$ is a key quantity in the spillover effects in Theorem \ref{thm:gen}, as having zero compliers would lead to zero spillover effects. This is similar in spirit to the non-interference case in a CRT where the compliers play a pivotal role in defining the effect of the treatment receipt \citep{angrist_identification_1996}. Under interference, the compliers also drive the effect of the treatment receipt by being \emph{influencers} among their peers and changing the number of compliers may enlarge the contrasts $k_{1j}$ and $k_{0j}$ in the estimands under a CRT. In contrast, the number of always-takers $n_j^{\rm AT}$ provide a ``baseline'' for the contrast in the spillover or the total effects and changing it does not change the said contrasts.

Corollary \ref{cor:1} states the form of $\tau$ under one-sided noncompliance in assumption (A4.1).
\begin{corollary} \label{cor:1} Suppose the assumptions in Theorem \ref{thm:gen} hold except we replace (A4) with (A4.1). Then,  $\tau$ becomes
\begin{align} \label{eq:tau_1sided}
\tau &=  \overline{\rm TE}^{\rm CO}(1,\mathbf{k}_{1} - \mathbf{1}; 0,\mathbf{k}_{0})  +  \frac{N^{\rm NT}}{N^{\rm CO}} \cdot \overline{\rm PE}^{\rm NT} (0,\mathbf{k}_{1}; 0,\mathbf{k}_{0}) 
\end{align}
where $\mathbf{k}_1 = (n_1^{\rm CO},\ldots,n_{J}^{\rm CO})$ and $\mathbf{k}_{0} = \mathbf{0}$.
\end{corollary}
We end the section by briefly discussing the relationship between instrument strength and our identification result. Broadly speaking, an instrument is strong when there are more compliers than non-compliers so that $N^{\rm NT} / N^{\rm CO} \approx 0$, $N^{\rm AT} / N^{\rm CO} \approx 0$ and $\tau_D$ is far away from $0$. When the instrument is strong, $\tau$ in Theorem \ref{thm:gen} mainly consists of compliers, specifically their total effects. However, if the instrument is weak and there are more never-takers or always-takers than compliers (i.e. $N^{\rm NT} / N^{\rm CO}$ or $N^{\rm AT} / N^{\rm CO}$ are far away from zero), $\tau$ predominantly represents the always-takers or never-takers in the population, and the effects among compliers may make up a small portion of $\tau$. However, this general observation won't hold when each subgroup has non-comparable magnitudes of treatment effects, such as if the compliers have substantially larger treatment effects than non-compliers.

\subsection{Unbiased Estimation of Network Effects is Generally Impossible} 
\label{sec:imp}

Thus far, we proved that the standard IV estimate is a mixture of network treatment effects. Next, we address whether any analysis with the observed data, not just the IV analysis in the prior section, can be informative about these network treatment effects. 

Formally, consider again the estimands defined in Section \ref{sec:target}, the total effect among compliers and the spillover effects among always-takers and never-takers. Unbiased estimators exist for the denominators of these quantities, which are the number of compliers, always-takers, and never-takers; see Section \ref{sec:est} for details. Hence, to learn more about total and spillover treatment effects from data, we need to be able to estimate the numerator of these quantities, which are the sums of individual total and spillover effects among their respective subgroups.
\begin{align} \label{eq:te_sum}
\overline{\rm TE}_{\rm sum}^{\rm CO}(1,\mathbf{k}_{1} - \mathbf{1}; 0,\mathbf{k}_{0})  &= \sum_{j=1}^{J} \sum_{i=1}^{n_j} {\rm TE}_{ji}(1,n_{j}^{\rm AT} + n_{j}^{\rm CO} -1;0,n_{j}^{\rm AT}) I(\text{$ji$ is CO}) \\
\overline{\rm PE}_{\rm sum}^{\rm AT}(1,\mathbf{k}_{1} - \mathbf{1}; 1,\mathbf{k}_{0} - \mathbf{1}) &=\sum_{j=1}^{J}  \sum_{i=1}^{n_j}  {\rm PE}_{ji}(1,n_{j}^{\rm AT} + n_{j}^{\rm CO} -1;1,n_{j}^{\rm AT} - 1) I(\text{$ji$ is AT})  \\
\overline{\rm PE}_{\rm sum}^{\rm NT} (0,\mathbf{k}_{1}; 0,\mathbf{k}_{0})  &= \sum_{j=1}^{J} \sum_{i=1}^{n_j} {\rm PE}_{ji}(0,n_{j}^{\rm AT} + n_{j}^{\rm CO};0,n_{j}^{\rm AT}) I(\text{$ji$ is NT}) \label{eq:pe0_sum} 
\end{align}
Here, we denote $\mathbf{k}_1 = (n_{1}^{\rm AT} + n_{1}^{\rm CO},\ldots,n_{J}^{\rm AT} + n_{J}^{\rm CO})$ and $\mathbf{k}_{0} = (n_{1}^{\rm AT},\ldots,n_{J}^{\rm AT})$. We focus on these specific forms of $\mathbf{k}_1$ and $\mathbf{k}_0$ because for a given study, they are the only ones that can be observed. Theorem \ref{thm:imp} shows that no unbiased estimators exist for these sums in CRTs with noncompliance and interference.
\begin{theorem} \label{thm:imp} Suppose the assumptions in Theorem \ref{thm:gen} or Corollary \ref{cor:1} hold. Then, there does not exist unbiased estimators for sums in equations \eqref{eq:te_sum}-\eqref{eq:pe0_sum} from the observed data $\mathbf{Y}, \mathbf{D}$, and $\mathbf{Z}$.
\end{theorem} 
Broadly speaking, the result of Theorem \ref{thm:imp} comes from the fact that for an unbiased estimator $T$ to exist for the three local estimands, the unbiased estimator must be able to classify all individuals into the three compliance types only based on the observed data. But, this is impossible, since it requires observing both potential outcomes $D_{ji}^{(1)}$ and $D_{ji}^{(0)}$. In contrast, without interference, the unbiased estimator for the CACE does not need to classify individuals, since we are only able to estimate the compliers' effect; specifically, without interference, only the compliers' effect remains in the ITT effect \citep{angrist_identification_1996}. The proof of Theorem \ref{thm:imp} states this more precisely, stating that an unbiased estimator, if it were to exist, must be able to differentiate individuals between different compliance types from the observed data, which is not possible with all values of the parameter space.  We remark that a variation of Theorem 1 in \citet{basse2018limitations} can also be used prove our Theorem \ref{thm:imp}.

As a negative result, Theorem \ref{thm:imp} highlights that traditional CRTs have fundamental limitations as experimental designs to elucidate network causal effects under noncompliance and interference and a need for either alternative experimental designs or stronger structural assumptions. The result is similar in spirit to negative results in \citet{manski1993identification},\citet{shalizi2011homophily}, \citet{manski2013identification}, and \citet{basse2018limitations} in that additional assumptions are needed to estimate effects under noncompliance and interference. Indeed, even with stratified interference (A5), which can be argued as a strong assumption primarily to estimate variances of estimators, we are unable to identify the sums of the total or spillover effects. We remark that our result differs from classic results in the IV literature that states that moments do not exist for ratio estimators since our focus in Theorem \ref{thm:imp} is on the numerator;  see Lemma 5 of \citet{nelson1990some} and references therein.

Finally, while Theorem \ref{thm:imp} proves that total and spillover effects are generally impossible to unbiasedly estimate under assumptions (A1)-(A5), they do not rule of local conditions that allow unbiased estimation.  For example, if all the subjects  are compliers, then we can estimate $\overline{\rm TE}^{\rm CO}(1;0)$ from data by using the standard IV method since the terms that are associated with non-compliers and that make up $\tau$ in Theorem \ref{thm:gen} go away. Similarly, under one-sided noncompliance, when all subjects are never-takers, we can estimate $\overline{\rm PE}^{\rm NT}(0)$ by using the standard IV method. These are two specific, local data generating scenarios that allow identification. However, in practice,  such scenarios are rare.

\subsection{Nonparametric Bounds under the Assumption of Non-Negative Treatment Effects} 
\label{sec:bound}
The results in Sections \ref{sec:ident} and \ref{sec:imp} imply analyzing a CRT with noncompliance and interference for network treatment effects is hopeless and one should focus on ITT effects. This pessimism is not entirely warranted and we show that in such CRTs with one-sided noncompliance, we can use one additional assumption to calculate informative bounds on the total effect among compliers and the spillover effect among never-takers.  

The assumption we invoke is one of non-negative treatment effects. Formally, a treatment is not harmful, or has non-negative total and spillover effects, if the following hold.
\begin{itemize}
\item[(A6)] \emph{Non-Negative Treatment Effects}. For any $0 \leq  k_{0} <  k_1 \leq n_j-1$ and individual $i$ in cluster $j$, we have
\begin{equation} 
0 \leq {\rm TE}_{ji}(1,k_1;0,k_0), \quad{} 0\leq {\rm PE}_{ji}(0,k_1; 0, k_0)
\end{equation}
\end{itemize}
We note that assumption (A6) is similar to an assumption used by \citet{choi_estimation_2017} about treatment monotonicity. In the PSDP, assumption (A6) is reasonable because it is unlikely that exposure to deworming treatments would increase the presence of infection. 

Let $y = \overline{\rm TE}^{\rm CO}(1,\mathbf{k}_1;0, \mathbf{k}_0)$ and $x = \overline{\rm PE}^{\rm NT}(0,\mathbf{k}_1; 0, \mathbf{k}_0)$. Then, under one-sided compliance, we can write $\tau$ in equation \eqref{eq:tau_1sided} as
\begin{equation} \label{eq:y-x}
y = \tau - \frac{N^{\rm NT}}{N^{\rm CO}} x
\end{equation}
Equation \eqref{eq:y-x} suggests an inverse relationship between the total average effect among compliers and the spillover average effect among never-takers. Specifically, as the spillover average effect becomes positive, the total average effect must become negative. Moreover, under (A6), $0\leq y $ and $0 \leq x$, which, along with equation \eqref{eq:y-x}, gives us bounds on both the total and the spillover effect.
\begin{equation} \label{eq:bound}
0 \leq  \overline{\rm TE}^{\rm CO}(1,\mathbf{k}_1;0, \mathbf{k}_0) \leq \tau, \quad{} 0\leq \overline{\rm PE}^{\rm NT}(0,\mathbf{k}_1; 0, \mathbf{k}_0)\leq \tau \frac{N^{\rm CO}}{N^{\rm NT}} 
\end{equation}
The boundaries of the bounds in equation \eqref{eq:bound} are achieved when either effect is at the extreme. For example, $\overline{\rm TE}^{\rm CO}(1,\mathbf{k}_1;0, \mathbf{k}_0)$ reaches its lower bound of $0$ if $\overline{\rm PE}^{\rm NT}(0,\mathbf{k}_1; 0, \mathbf{k}_0) = \tau \frac{N^{\rm CO}}{N^{\rm NT}}$ and $\overline{\rm PE}^{\rm NT}(0,\mathbf{k}_1; 0, \mathbf{k}_0)$ reaches its lower bound of $0$ if $\overline{\rm TE}^{\rm CO}(1,\mathbf{k}_1;0, \mathbf{k}_0)= \tau$. 

If the outcome is binary, we can impose tighter constraints on both equations \eqref{eq:y-x} and \eqref{eq:bound}
\begin{equation} 
\label{eq:y-x2}
y = \tau - \frac{N^{\rm NT}}{N^{\rm CO}} x, \quad{} 0\leq y \leq 1, 0 \leq x \leq 1
\end{equation}
\noindent which leads to
\begin{align} \label{eq:bound2}
\max\left(0,\tau - \frac{N^{\rm NT}}{N^{\rm CO}} \right) &\leq  \overline{\rm TE}^{\rm CO}(1,\mathbf{k}_1;0, \mathbf{k}_0) \leq \min(1,\tau) \\ 
\max\left(0,\frac{N^{\rm CO}}{N^{\rm NT}} (\tau-1)\right)&\leq\overline{\rm PE}^{\rm NT}(0,\mathbf{k}_1; 0, \mathbf{k}_0) \leq  \min\left(1, \tau \frac{N^{\rm CO}}{N^{\rm NT}} \right) \label{eq:bound2.1}
\end{align}
\noindent In equation \eqref{eq:bound2.1}, if the number of compliers exceeds the number of never-takers, this tightens the lower bound for the total effect. However, if the number of never-takers increases and exceeds the number of compliers, we have a tighter upper bound for the spillover effect. This is consistent with Theorem \ref{thm:gen} where $\tau$ is a weighted mixture of subgroup effects and the weights are proportional to the number of individuals in each subgroup.

Figure \ref{fig:1} summarizes the characteristics of the bounds in equations \eqref{eq:bound2} and \eqref{eq:bound2.1} by plotting the range of spillover and total effects as a function of the proportion of compliers $N^{\rm CO}/N$ and $\tau$.
\begin{figure}[h!]
\begin{center}
\includegraphics[width=\textwidth]{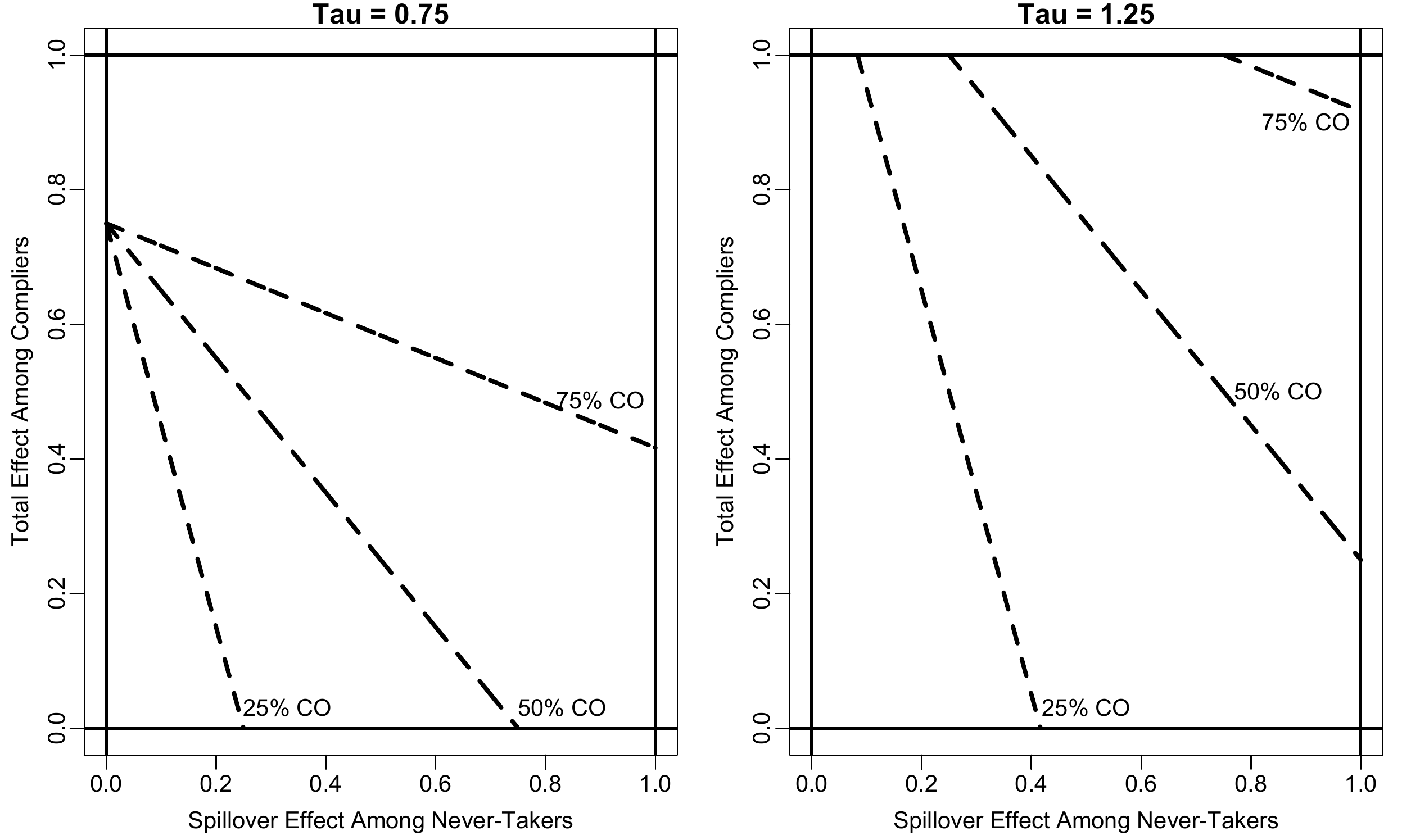}
\end{center}
\caption{A Numerical Example of Bounding Total Effect Among Compliers and Spillover Effect Among Never-Takers in a One-Sided Noncompliance CRT with Binary Outcomes.}
\label{fig:1}
\end{figure}
For example, when $\tau = 0.75$, as the number of compliers increase, the plausible range of the total effect among compliers decreases. However, as the number of compliers decrease, $\tau$ becomes more representative of the never-takers and hence, the range of plausible values for the spillover effect is shorter and range of plausible values for the total effect is wider. Note that the upper bound on the total effect is $0.75$ and it has a non-trivial (i.e. non-zero) lower-bounded if the compliance rate is high. 

When $\tau = 1.25$, the spillover effect among never-takers is always bounded away from zero; more generally, if $\tau > 1$, the lower bound on the spillover effect is always bounded away from zero. The upper bound depends on the compliance rate. For example, if the compliance rate is high, say 75\%, there are tight bounds for both the spillover and the total effect. However, if the compliance rate is 50\%, both the spillover and total effect have wide bounds, although neither include zero. Finally, if the compliance rate is 25\%, the bound on the total effect is not informative, i.e. it ranges from $0$ to $1$, but the bound on the spillover effect becomes narrower.

As a general rule of thumb for binary $Y$, the numerical example suggests that for CRTs with one-sided noncompliance where assumption (A6) is plausible, informative analysis is possible for spillover and total treatment effects. Specifically, with $\tau >1 $ and a high compliance rate, we can obtain tight bounds on both the total and spillover effects. If $\tau  > 1$ and the compliance rate is low, then we can obtain tight bounds on the spillover effect, but not the total effect. If $\tau < 1$ and the compliance rate is high, then we can also obtain tight bounds on the total effect, but not the spillover effect. If $\tau <1 $ and the compliance rate is low, then we can only obtain tight upper bounds on both the total and spillover effects.

\section{Estimation and Inference of Bounds} \label{sec:est}
\subsection{Estimation}
We now outline how to estimate the bounds in Section~\ref{sec:bound} when the outcome is binary as it is in the PSDP. Generally, estimation of the three key values that make up the bounds, $\tau$, $N^{\rm CO}$, and $N^{\rm NT}$, directly follow from literature on IV estimation. To estimate $\tau$, we can use the ratio of the estimated ITT effect with the estimated compliance effect, i.e. 
\begin{align*}
\widehat{\tau}_{Y} &= \frac{1}{N} \left(\sum_{j=1}^{J} \frac{J}{m} Z_{j} \sum_{i=1}^{n_j} Y_{ji} -   \sum_{j=1}^{J} \frac{J}{J-m} (1 - Z_{j}) \sum_{i=1}^{n_j} Y_{ji}  \right), \quad{} \widehat{\tau}_D = \frac{1}{N} \sum_{j=1}^{J} \frac{J}{m} Z_{j} \sum_{i=1}^{n_j} D_{ji}  \\
\widehat{\tau} &= \frac{\widehat{\tau}_{Y}}{\widehat{\tau}_{D}}
\end{align*}
As mentioned earlier, $\widehat{\tau}_{D}$ is also an estimator of the compliance rate or the proportion of compliers and it ranges between $0$ and $1$. \citet{kangkeele2018} prove that this version of the ratio estimator $\hat{\tau}$ is superior in terms of its finite sample properties compared to other IV methods for CRTs; see \citet{kangkeele2018} for details. To estimate the subgroup sizes, $N^{\rm CO}$ and $N^{\rm NT}$, we use the following estimators 
\begin{align*}
\widehat{N}^{\rm CO} &= \sum_{j=1}^{J} \frac{J}{m} Z_j \sum_{i=1}^{n_j} D_{ji}, \quad{} \widehat{N}^{NT} = N - \widehat{N}^{\rm CO}
\end{align*}
Under assumptions (A1), these estimators are unbiased for $N^{\rm CO}$ and $N^{\rm NT}$.

Let $L_{\rm TE}^{\rm CO} =\max\left(0,\tau - \frac{N^{\rm NT}}{N^{\rm CO}} \right), U_{\rm TE}^{\rm CO} = \min(1,\tau)$, $L_{\rm PE}^{\rm NT} =  \max\left(0,\frac{N^{\rm CO}}{N^{\rm NT}} (\tau-1)\right)$, and $U_{\rm PE}^{\rm NT} =  \min\left(1, \tau \frac{N^{\rm CO}}{N^{\rm NT}} \right)$ be the lower and upper bounds of the total and spillover effects in equations \eqref{eq:bound2} and \eqref{eq:bound2.1}. Let $\widehat{L}_{\rm TE}^{\rm CO} =\max\left(0,\widehat{\tau} - \frac{\widehat{N}^{\rm NT}}{\widehat{N}^{\rm CO}} \right), \widehat{U}_{\rm TE}^{\rm CO} = \min(1,\widehat{\tau})$, $\widehat{L}_{\rm PE}^{\rm NT} =  \max\left(0,\frac{\widehat{N}^{\rm CO}}{\widehat{N}^{\rm NT}} (\widehat{\tau}-1)\right)$, and $\widehat{U}_{\rm PE}^{\rm NT} =  \min\left(1, \widehat{\tau} \frac{\widehat{N}^{\rm CO}}{\widehat{N}^{\rm NT}} \right)$ be the plug-in estimates for these bounds. The following theorem shows that these plug-in estimators are consistent estimators under the asymptotic regime where the cluster size remains fixed while the number of clusters go to infinity.
\begin{theorem} \label{thm:bound_est} Suppose assumptions (A1)-(A4.1) and (A5)-(A6) hold and we have binary outcomes. 
Consider the asymptotic regime where the number of clusters go to infinity while the cluster size remains bounded, i.e. $J,J-m \to \infty$ where (i) $m/J \to p \in (0,1)$, (ii) $n_j \leq B$ for some constant $B$, (iii) $\tau_D$ and $\tau_Y$ are fixed for all $J$. Then, 
\begin{align*}
| \widehat{L}_{\rm TE}^{\rm CO} - L_{\rm TE}^{\rm CO} | &\to 0, \quad{} | \widehat{U}_{\rm TE}^{\rm CO} - U_{\rm TE}^{\rm CO} | \to 0, \quad{} |\widehat{L}_{\rm PE}^{\rm NT} - L_{\rm PE}^{\rm NT}| \to 0, \quad{} | \widehat{U}_{\rm PE}^{\rm NT} - U_{\rm PE}^{\rm NT} | \to 0 
\end{align*}
in probability.
\end{theorem}
Conditions (i) and (ii) in Theorem \ref{thm:bound_est} are standard asymptotic regimes for CRTs where the number of clusters go to infinity while the cluster size remains fixed; \citet{kangkeele2018} shows that in the opposite asymptotic regime where the number of clusters remain fixed, but the cluster size goes to infinity, typical CRT estimators exhibit poor properties. Condition (iii) follows Chapter 4.4 of \citet{lehmann_elements_2004} where in finite sample settings, the asymptotic embedding sequence has the same mean for every $N$; this type of asymptotics has been used in noncompliance settings \citep{baiocchi_building_2010}.  Condition (iii) can be generalized to a condition where $\tau_Y$ and $\tau_D$ converge as sample size increases; in this case, $\tau_D$ must be bounded strictly away from $0$ and $1$. 
We will use Theorem \ref{thm:bound_est} as a basis for constructing bounds in our empirical example in Section \ref{sec:app}.

\subsection{Simultaneous Confidence Interval for Bounds}
While Theorem \ref{thm:bound_est} provides a consistent estimator for the bounds, it does not characterize uncertainty. Indeed, the presence of min or max operators in the bounds makes it difficult to derive a closed form expression for the asymptotic distribution of bounds and bootstrap techniques have been proven to be generally invalid in this context\citep{romano1989bootstrap, andrews2000inconsistency,romano2008inference,romano2010inference, andrews2009validity,hirano2012impossibility}; see \citet{tamer2010partial} for a general overview. One solution is to assume an infinite population model with independent and identically distributed (IID) observables. Under an IID infinite population model, \citet{cheng_bounds_2006} used the percentile bootstrap method of \citet{horowitz_nonparametric_2000} to construct confidence intervals for bounds in an IV setting. However, in our application, the assumption of an IID infinite population model is unrealistic given the clustered treatment assignment and interference.

One alternative would be to use the methods in \citet{chernozhukov2013intersection}. They developed  methods for asymptotically bias-corrected estimators and confidence intervals for intersection bounds under more general sampling mechanisms. The Chernozhukov-Lee-Rosen (CLR) approach uses precision-corrected estimates of the terms in the bounding functions before applying the min and max operators. However, using the CLR approach requires a non-trivial coupling argument between our plug-in estimators for bounds and a suitably simple and uniformly concentrating Gaussian process to conduct inference, and is difficult to verify.

Instead, we propose a finite-sample randomization inference method to test the joint hypothesis for the bounds, denoted as $y$ and the spillover effect, denoted as $x$, and to invert this test to obtain simultaneous confidence intervals for both sets of bounds. Formally, consider testing the joint hypothesis
\begin{equation} \label{eq:hyp_bound}
H_{0}: x = x_0, y =y_0
\end{equation}
From Corollary \ref{cor:1}, the null hypothesis $H_0$ in \eqref{eq:hyp_bound} implies the following hypothesis about $\tau$
\begin{equation} \label{eq:hyp_tau}
H_{0}': \tau = \tau_{x_0, y_0}, \quad{} \tau_{x_0,y_0} = y_0 + \frac{p^{\rm NT}}{p^{\rm CO}} x_0
\end{equation}
where $p^{\rm CO} = N^{\rm CO} / N = p^{\rm CO} \in (0,1)$ and $p^{\rm NT} = N^{\rm NT} / N = p^{\rm NT} \in (0,1)$ are the population proportion of complier and never-takers, respectively. If $p^{\rm NT}$ and $p^{\rm CO}$ are known, \citet{kangkeele2018} provided a method to test $H_0'$ by using the difference in adjusted outcomes. Specifically, for each cluster $j$, let $Y_{j} = \sum_{i=1}^{n_j} Y_{ji}$ and $D_j = \sum_{i=1}^{j} D_{ji}$ be the sums of $Y$ and $D$ respectively. Also, let $A_j(\tau_{x_0,y_0}) = Y_{j} - \tau_{x_0,y_0} D_{j}$ be the  adjusted outcome of $Y_{j}$ by $D_{j}$ and let $A_T = \sum_{j=1}^{J} Z_j A_j(\tau_{x_0,y_0}) / m$, and $A_{C}(\tau_{x_0,y_0})= \sum_{j=1}^{J} (1 - Z_j) A_j(\tau_{x_0,y_0}) / (J-m)$ be the means of these adjusted outcomes for the treated and control clusters, respectively. Then, \citet{kangkeele2018} proposed the following test statistic for $H_0'$ 
\begin{equation} \label{eq:test_stat_t}
T(\tau_{x_0,y_0}) = \frac{ A_{T}(\tau_{x_0,y_0}) -A_{C}(\tau_{x_0,y_0})}{\sqrt{ {\rm Var} \left[A_{T}(\tau_{x_0,y_0}) -A_{C}(\tau_{x_0,y_0}) \mid \mathcal{F}, \mathcal{Z} \right]}}
\end{equation}
for some suitable estimator of the variance. For instance, a popular estimator for the variance is
\begin{align*}
\widehat{{\rm Var}}  \left[A_{T}(\tau_{x_0,y_0}) -A_{C}(\tau_{x_0,y_0}) \mid \mathcal{F}, \mathcal{Z} \right] =& \frac{1}{m(m-1)} \sum_{j=1}^{J} Z_j \left\{ A_j(\tau_{x_0,y_0}) -  A_{T}(\tau_{x_0,y_0})\right\}^2 \\
&\quad{}+ \frac{1}{(J-m)(J-m-1)} \sum_{j=1}^{J} (1 - Z_j) \left\{ A_j(\tau_{x_0,y_0}) -  A_{C}(\tau_{x_0,y_0})\right\}^2
\end{align*}
If $p^{\rm NT}$ and $p^{\rm CO}$ are unknown, we can plug-in consistent estimators of them, say $\widehat{p}^{\rm NT} = \widehat{N}^{\rm NT}/ N$ and $\widehat{p}^{\rm CO} = \widehat{N}^{\rm CO} / N$, respectively, from Theorem \ref{thm:bound_est} and use $\widehat{\tau}_{x_0,y_0} = y_0 + \widehat{p}^{\rm NT}/ \widehat{p}^{\rm CO} x_0$ in lieu of $\tau_{x_0,y_0}$. The following Theorem shows that with a plug-in estimate, the test statistic $T(\widehat{\tau}_{x_0,y_0})$ is asymptotically pivotal under $H_0$.
\begin{theorem} \label{thm:bound_CI} Consider the assumptions in Theorem \ref{thm:bound_est} and suppose $H_0$ holds. As $J$ grows, consider a sequence of $\mathcal{F}$ where $\tau_Y$ and $\tau_D$ are constants so that $\tau_Y/ \tau_D = y_0 + p^{\rm NT}/p^{\rm CO} x_0$ is a constant, the test statistic $T(\widehat{\tau}_{x_0,y_0})$ is asymptotically pivotal and Normally distributed, i.e.
\[
T(\widehat{\tau}_{x_0,y_0}) \to N(0,1)
\]
in distribution.
\end{theorem}
Theorem \ref{thm:bound_est} provides a basis to obtain $1-\alpha$ simultaneous confidence intervals for $x$ and $y$ which represent the total and spillover effects. Specifically, for any $\alpha \in (0,1)$, the $1-\alpha$ confidence set $\mathcal{C}_{1-\alpha}$ of the two effects is the set of values of $x$ and $y$ where
\[
\mathcal{C}_{1-\alpha} = \{ (x_0, y_0) \mid 0\leq x_0\leq 1, 0 \leq y_0 \leq 1, |T(\hat{\tau}_{x_0,y_0})| \leq z_{1-\alpha/2} \}
\]
Here, $z_{1-\alpha.2}$ is the $1-\alpha/2$ quantile of the standard Normal. By the duality of testing and confidence interval, this confidence set covers a \emph{fixed} true value of $x$ and $y$ with at least $1-\alpha$ probability; see \citet{imbens_confidence_2004} for discussion on interpreting confidence intervals for bounds. Additionally, while a grid search over $x_0$ to $y_0$ to find $\mathcal{C}_{1-\alpha}$ is feasible since the outcome is binary, \citet{kangkeele2018} shows that the confidence interval for $\tau$ in $H_0'$ can be efficiently solved using a quadratic equation. Consequently, since $\tau$ is a linear w.r.t. $x_0$ and $y_0$, the joint confidence interval for $x_0$ and $y_0$ is also a solution to a quadratic equation and $\mathcal{C}_{1-\alpha}$ can be efficiently solved.  This computational advantage is in contrast to many inferential methods under the partial identification literature \citep{romano2008inference,romano2010inference,chernozhukov2013intersection} where expensive computation is needed. 

A downside of the proposed approach is that the interval may be conservative. In particular, if the true value of $x$ and $y$ satisfy the equation $y + p^{\rm NT}/p^{\rm CO} x = y_0 + p^{\rm NT}/p^{\rm CO} x_0$, then the test $T(\tau_{x_0,y_0})$ will have no power to reject $H_0'$, leading to potentially large confidence regions. Nevertheless, the confidence set $\mathcal{C}_{1-\alpha}$ will still have Type I error control. 

\section{Application: The Effect of Deworming in the Presence of Noncompliance and Spillovers}
\label{sec:app}
This section revisits the statistical analysis of the PSDP intervention under the new paradigm where we allow for interference and noncompliance. The initial analysis of the PSDP intervention focused on estimating ITT effects, specifically the direct effect of the deworming intervention and the treatment spillover effects \citep{miguel2004worms}; the original analysis did not include any analyses that focused on complier effects. The original data did, however, include detailed measures of whether students complied with cluster-level treatment assignment. In our analysis, we use an updated version of the data that corrected a series of data errors in the original data \citep{aiken2015re,hicks2015commentary}. We focus on the primary outcome from the original study, which was a binary indicator for the presence of a helminth infection with $1$ denoting infection and $0$ otherwise. We note that in this example, taking the treatment (new deworming intervention) leads to a decrease in outcome (i.e. infection) and hence, our effect estimates are negative. Also, we remind readers that the PSDP intervention had one-sided noncompliance where individuals in control clusters could not seek out the new medication.

First, we carry out the standard IV analysis, which is what an investigator may naively try with CRT data with noncompliance, falsely assuming that effects from treatment spillovers are mitigated due to clustered treatment assignment. Using the methods in \citet{kangkeele2018} to estimate $\tau$, the estimate of $\tau$ is $-0.79$ with a 95\% confidence interval of $-1$ and $-0.51$. While $\tau$ is easily estimable using existing IV methods, this estimate of $\tau$ is no longer the effect for compliers alone, but also includes effects among never-takers. 

Next, we focus on estimating the bounds for both the total effect among compliers, $\overline{\rm TE}^{\rm CO}(1;0)$, and the spillover effect among never-takers, $\overline{\rm PE}^{\rm NT}(0)$. We can plug in estimates of $\widehat{N}^{\rm NT}, \widehat{N}^{\rm CO}$, and $\hat{\tau}$ into equation \eqref{eq:bound} to estimate the bounds. In the PDSP intervention, the instrument is fairly strong: average student compliance in schools assigned to treatment was 60\%. As such, the number of compliers is higher compared to the number of never-takers, which indicates we should be able to obtain relatively informative bounds on $\overline{\rm TE}^{\rm CO}(1;0)$, but not on $\overline{\rm PE}^{\rm NT}(0;0)$. Indeed, when we plug in estimates for our bounds, we obtain estimates of $-0.79$ and $-0.12$ for the total effect among compliers and $0$ and $1$ for the spillover effect among never-takers. Note that by construction, the lower bound is equivalent to the estimate for $\tau$. Next, we computed the 95\% confidence interval for the total effect and the spillover effect using the proposed method. The confidence intervals are $-1$ and $0$; in words, once the uncertainty of the bound estimates are taken into considering, the bounds for the total and spillover effects span the entire parameter space. We also estimated the confidence intervals using the bootstrap and the methods outlined in \citet{chernozhukov2013intersection}. Both methods returned a confidence region of -1 and 0. 

Figure \ref{fig:ci_band} graphically represents the results of the bounds from the study. Each point inside the shaded area represents plausible values of the total and spillover effect after taking into account the uncertainty in estimating the bounds. From the plot, we see that certain combinations of total effect and spillover effect values are improbable. For example, based on this study with the given sample size, it's unlikely that the deworming treatment led to 75\% reduction in helminth infections among both compliers and never-takers, as measured by the total effect among compliers and the spillover effect among never-takers. Similarly, it's unlikely deworming treatment lead to only 25\% reduction in helminth infections among both compliers and never-takers.

\begin{figure}
\centering
\includegraphics[scale=.8]{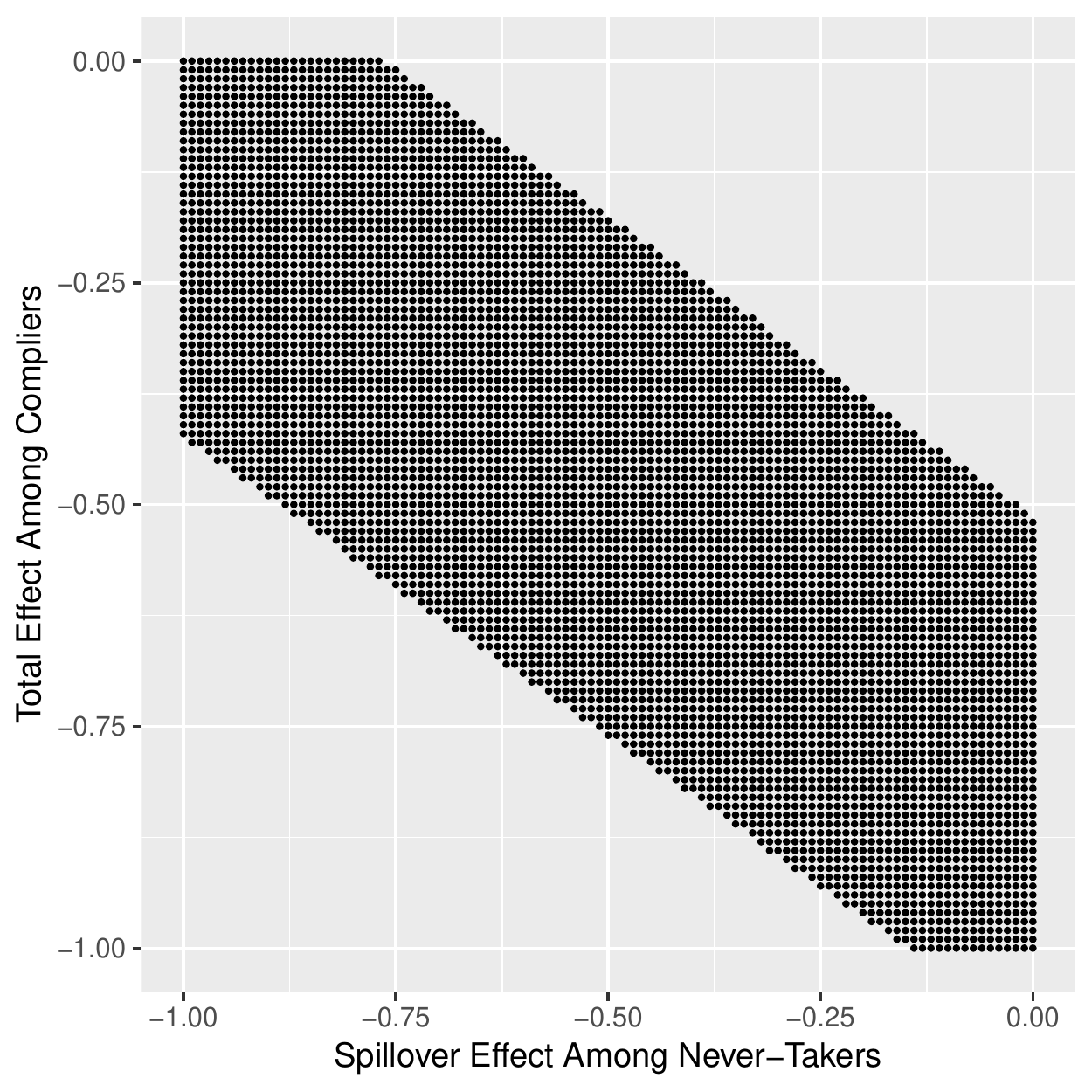}
\caption{95\% Simultaneous Confidence Intervals for the Total Effect Among Compliers and the Spillover Effect Among Never-Takers of the New Deworming Treatment. Here, taking the treatment (i.e. new deworming medication) will lead to a smaller outcome (i.e. the presence of helminth infections) and hence, our effect estimates are negative in sign.  Each point inside the shaded area represents plausible values of the total and spillover effect after taking into account the uncertainty in estimating the bounds.}
\label{fig:ci_band}
\end{figure}

\section{Conclusion}

In this paper, we studied CRTs with both noncompliance and treatment spillovers. In many public health interventions, subjects may refuse treatments but are partially exposed when other subjects take the treatment. In the PSDP, some students did not take the medications that comprised the treatment. However, infection levels may be lower for these unexposed students as peers who took the treatment lowered their likelihood of environmental exposure to helminths.

We showed that standard causal estimands of interest cannot be estimated under the usual assumptions used with noncompliance and interference. The standard IV analysis leads to a mixture of causal effects instead of the usual complier average causal effect. We extended the result to show that unbiased estimation of key components of the causal estimands in Section \ref{sec:target} is impossible. Finally, we showed that investigators must rely on partial identification methods to place bounds on these quantities. While partial identification results often produce wide bounds that are uninformative, we showed that in the PSDP data and in Figure \ref{fig:ci_band}, the bounds can be informative even once sampling uncertainty is accounted for.

\appendix

\section{Proofs of Key Theorems}

\begin{proof}[Proof of Theorem \ref{thm:gen}] We start by decomposing the term $Y_{ji}^{(1,D_{ji}^{(1)},\mathbf{D}_{j-i}^{(1)})} - Y_{ji}^{(0,D_{ji}^{(0)}, \mathbf{D}_{j-i}^{(0)})}$ in $\tau_Y$. Under network exclusion restriction (A3) and monotonicity (A4), we have
{\small
\begin{align*}
&Y_{ji}^{(1,D_{ji}^{(1)},\mathbf{D}_{j-i}^{(1)})} - Y_{ji}^{(0,D_{ji}^{(0)}, \mathbf{D}_{j-i}^{(0)})} \\
=&Y_{ji}^{(D_{ji}^{(1)},\mathbf{D}_{j-i}^{(1)})} - Y_{ji}^{(D_{ji}^{(0)}, \mathbf{D}_{j-i}^{(0)})} \\
=& \left[Y_{ji}^{(1,\mathbf{D}_{j-i}^{(1)})} - Y_{ji}^{(1, \mathbf{D}_{j-i}^{(0)})}\right] I(\text{\rm $ji$ is AT}) + \left[Y_{ji}^{(0,\mathbf{D}_{j-i}^{(1)})} - Y_{ji}^{(0, \mathbf{D}_{j-i}^{(0)})}\right] I(\text{\rm $ji$ is NT}) + \left[Y_{ji}^{(1,\mathbf{D}_{j-i}^{(1)})} - Y_{ji}^{(0, \mathbf{D}_{j-i}^{(0)})}\right] I(\text{\rm $ji$ is CO}) \\
=& {\rm PE}_{ji}(1,\mathbf{D}_{j-i}^{(1)}; 1, \mathbf{D}_{j-i}^{(0)}) I(\text{\rm $ji$ is AT}) + {\rm PE}_{ji}(0,\mathbf{D}_{j-i}^{(1)}; 0, \mathbf{D}_{j-i}^{(0)}) I(\text{\rm $ji$ is NT}) + {\rm TE}_{ji}(1,\mathbf{D}_{j-i}^{(1)};0,\mathbf{D}_{j-i}^{(0)})  I(\text{\rm $ji$ is CO}) 
\end{align*}
}
where 
\begin{align*}
{\rm PE}_{ji}(d_{ji},\mathbf{D}_{j-i}^{(1)}; d_{ji}, \mathbf{D}_{j-i}^{(0)}) &= Y_{ji}^{(d_{ji},\mathbf{D}_{j-i}^{(1)})} - Y_{ji}^{(d_{ji}, \mathbf{D}_{j-i}^{(0)})} \\
{\rm TE}_{ji}(1,\mathbf{D}_{j-i}^{(1)};0,\mathbf{D}_{j-i}^{(0)}) &= Y_{ji}^{(1,\mathbf{D}_{j-i}^{(1)})} - Y_{ji}^{(0, \mathbf{D}_{j-i}^{(0)})}
\end{align*}
Notice that so long as the vectors $\mathbf{D}_{j-i}^{(1)} - \mathbf{D}_{j-i}^{(0)} \neq \mathbf{0}$, we can guarantee that the spillover effects for ATs and NTs do not disappear. This occurs when there is at least one CO in cluster $j$, say individual $i'$, because $D_{ji'}^{(1)} - D_{ji'}^{(0)} = 1$. Then, summing these quantities over all clusters give you the sum of total effects and spillover effects for different subgroups CO, AT, and NT, all without (A5).
\\\\
With (A5), ${\rm PE}_{ji}(d_{ji},\mathbf{D}_{j-i}^{(1)}; d_{ji}, \mathbf{D}_{j-i}^{(0)})$ and ${\rm TE}_{ji}(1,\mathbf{D}_{j-i}^{(1)};0,\mathbf{D}_{j-i}^{(0)})$ simplify to the sum of the vectors $\mathbf{D}_{j-i}^{(z)}$ for different $z \in \{0,1\}$. In particular, the sum of the vector $\mathbf{D}_{j-i}^{(1)}$, i.e. $\sum_{i' \neq i} D_{ji'}^{(1)}$, is the number of compliers and always-takers in cluster $j$, minus 1 if $i$ is either a complier or an always-taker or $0$ otherwise. The sum of the vector $\mathbf{D}_{j-i}^{(0)}$, i.e. $\sum_{i' \neq i} D_{ji'}^{(0)}$, is the number of always-takers in cluster $j$ minus $1$ if $i$ is an always-taker or $0$ otherwise.
\\\\
In both cases, dividing this quantity $\tau_D$, which constitutes the total number of compliers in the population and is non-zero because of (A2), gives the desired result. A similar argument can be used to prove the decomposition under one-sided noncompliance where we remove ATs in the above expression.
\end{proof}

\begin{proof}[Proof of Theorem \ref{thm:imp}]
Let $\mathcal{F} = \{Y_{ji}^{(D_{ji}^{(z)}, \mathbf{D}_{j-i}^{(z)})}, D_{ji}^{(z)} \mid j=1,\ldots,J,i=1,\ldots,n_j, z\in \{0,1\} \}$ be the parameter set satisfying (A2)-(A4). By knowing $\mathcal{F}$, we can fully characterize the distribution of the observed data $Y_{ji}, D_{ji}$ and $Z_j$. An unbiased estimator $S(\mathbf{Y}, \mathbf{D}, \mathbf{Z})$ of an estimand $\theta(\mathcal{F})$, say the sum of individual total effect among compliers, satisfies
\begin{equation} \label{eq:unbias}
\theta(\mathcal{F}) = E[T(\mathbf{Y}, \mathbf{D}, \mathbf{Z}) \mid \mathcal{F}, \mathcal{Z}]
\end{equation}
for all parameter values. We show that for the estimands in Theorem \ref{thm:imp}, such a function $T$ does not exist. 

First, we work with $\overline{\rm TE}_{\rm sum}^{\rm CO}(1, \mathbf{k}_{1}-\mathbf{1}; 0,\mathbf{k}_{0}; \mathcal{F})$, where we add the argument $\mathcal{F}$ for clarity. Without loss of generality, let the first cluster $j=1$ contain at least two individuals. Suppose by contradiction, there is an unbiased estimator $T(\mathbf{Y}, \mathbf{D}, \mathbf{Z})$ for $\overline{\rm TE}_{\rm sum}^{\rm CO}(1, \mathbf{k}_{1}-\mathbf{1}; 0,\mathbf{k}_{0}; \mathcal{F})$ and consider two parameter sets, $\mathcal{F}$ and $\mathcal{F}'$. Parameter set $\mathcal{F}$ has at least one always-taker, say $i=1$, and one complier, say $i=2$, in cluster $j=1$. Parameter $\mathcal{F}'$ flips the always-taker $i=1$ in $\mathcal{F}$ to a complier and the complier $i=2$ in $\mathcal{F}$ to an always-taker. The rest remains identical between the two parameters $\mathcal{F}$ and $\mathcal{F}'$, in particular the number of always-takers and compliers. Then, the difference between the two estimands under two parameters $\mathcal{F}, \mathcal{F}'$ is 
\begin{align*}
&\overline{\rm TE}_{\rm sum}^{\rm CO}(1, \mathbf{k}_{1}-\mathbf{1}; 0,\mathbf{k}_{0}; \mathcal{F}') - \overline{\rm TE}_{\rm sum}^{\rm CO}(1, \mathbf{k}_{1}-\mathbf{1}; 0,\mathbf{k}_{0}; \mathcal{F}) \\
=& \left(Y_{12}^{(1,[n_{j}^{\rm AT} + n_{j}^{\rm CO}-1])} - Y_{12}^{(0,[n_{j}^{\rm AT}])}  \right) - \left(Y_{11}^{(1,[n_{j}^{\rm AT} + n_{j}^{\rm CO}-1])} - Y_{11}^{(0,[n_{j}^{\rm AT}])} \right)
\end{align*}
Also, by the definition of expectations and assumption (A1), for any $\mathcal{F}$, $T(\mathbf{Y}, \mathbf{D}, \mathbf{Z})$ must obey
\begin{align*}
&E[T(\mathbf{Y}, \mathbf{D}, \mathbf{Z}) \mid \mathcal{F}, \mathcal{Z}] \\
=& \frac{1}{{J \choose m}} \sum_{\mathbf{z} \in \mathcal{Z}} T(\mathbf{Y}, \mathbf{D},\mathbf{z}) \\
=& \frac{1}{{J \choose m}} \left[ \sum_{\mathbf{z} \in \mathcal{Z}; z_1 = 1} T(\mathbf{Y}, \mathbf{D},1,\mathbf{z}_{-1}) + \sum_{\mathbf{z} \in \mathcal{Z}; z_1 = 0} T(\mathbf{Y}, \mathbf{D},0,\mathbf{z}_{-1})  \right] \\
=& \frac{1}{{J \choose m}}  \left[ \sum_{\mathbf{z} \in \mathcal{Z}; z_1 = 1} T(Y_{11}^{(D_{11}^{(1)}, \mathbf{D}_{1-1}^{(1)})},\ldots, Y_{1n_1}^{(D_{1n_1}^{(1)}, \mathbf{D}_{1-n_1}^{(1)})},\ldots,Y_{Jn_J}^{(D_{Jn_J}^{(z_J)}, \mathbf{D}_{J-n_J}^{(1)})}, D_{11}^{(1)},\ldots,D_{1n_1}^{(1)},\ldots,D_{Jn_J}^{(z_J)},1,\mathbf{z}_{-1}) \right. \\
&+ \left. \sum_{\mathbf{z} \in \mathcal{Z}; z_1 = 0}  T(Y_{11}^{(D_{11}^{(0)}, \mathbf{D}_{1-1}^{(0)})},\ldots, Y_{1n_1}^{(D_{1n_1}^{(0)}, \mathbf{D}_{1-n_1}^{(0)})},\ldots,Y_{Jn_J}^{(D_{Jn_J}^{(z_J)}, \mathbf{D}_{J-n_J}^{(z_J)})}, D_{11}^{(0)},\ldots,D_{1n_1}^{(0)},\ldots,D_{Jn_J}^{(z_J)},0,\mathbf{z}_{-1}) \right] 
\end{align*}
Taking the difference between $E[T(\mathbf{Y}, \mathbf{D}, \mathbf{Z}) \mid \mathcal{F}', \mathcal{Z}]$ and $E[T(\mathbf{Y}, \mathbf{D}, \mathbf{Z}) \mid \mathcal{F}, \mathcal{Z}]$ leads to 
\begin{align*}
=& \frac{1}{{J \choose m}}  \left[ \sum_{\mathbf{z} \in \mathcal{Z}; z_1 = 1} T(Y_{11}^{(1, [n_{j}^{\rm AT} + n_{j}^{\rm CO}-1])},Y_{12}^{(1, [n_{j}^{\rm AT} + n_{j}^{\rm CO}-1])},\ldots, D_{11}^{(1)} = 1, D_{12}^{(1)} = 1, \ldots,z_{1} = 1,\mathbf{z}_{-1}) \right. \\
&+ \left. \sum_{\mathbf{z} \in \mathcal{Z}; z_1 = 0}  T(Y_{11}^{(1, [n_{j}^{\rm AT}-1])},Y_{12}^{(0, [n_{j}^{\rm AT}])},\ldots, D_{11}^{(0)} = 1, D_{12}^{(0)} = 0, \ldots,z_{1} = 0,\mathbf{z}_{-1}) \right] \\
&\quad{}- \frac{1}{{J \choose m}}  \left[ \sum_{\mathbf{z} \in \mathcal{Z}; z_1 = 1} T(Y_{11}^{(1, [n_{j}^{\rm AT} + n_{j}^{\rm CO}-1])},Y_{12}^{(1, [n_{j}^{\rm AT} + n_{j}^{\rm CO}-1])},\ldots, D_{11}^{(1)} = 1, D_{12}^{(1)} = 1, \ldots,z_{1} = 1,\mathbf{z}_{-1}) \right. \\
&+ \left. \sum_{\mathbf{z} \in \mathcal{Z}; z_1 = 0}  T(Y_{11}^{(0, [n_{j}^{\rm AT}])},Y_{12}^{(1, [n_{j}^{\rm AT} -1])},\ldots, D_{11}^{(0)} = 0, D_{12}^{(0)} = 1, \ldots,z_{1} = 0,\mathbf{z}_{-1}) \right] \\
=& \frac{1}{{J \choose m}}  \left[ \sum_{\mathbf{z} \in \mathcal{Z}; z_1 = 0}  T(Y_{11}^{(1, [n_{j}^{\rm AT} -1])},Y_{12}^{(0, [n_{j}^{\rm AT} ])},\ldots, D_{11}^{(0)} = 1, D_{12}^{(0)} = 0, \ldots,z_{1} = 0,\mathbf{z}_{-1}) \right. \\
&\quad{} \left. -  T(Y_{11}^{(0, [n_{j}^{\rm AT}])},Y_{12}^{(1, [n_{j}^{\rm AT} -1])},\ldots, D_{11}^{(0)} = 0, D_{12}^{(0)} = 1, \ldots,z_{1} = 0,\mathbf{z}_{-1})  \right]
\end{align*}
where we suppress arguments to $T$ that do not change. By our assumption about unbiasedness, the difference in expectation has to equal to the difference in the estimand between $\mathcal{F}'$ and $\mathcal{F}$. However, the estimand's difference has the terms $Y_{12}^{(1,[n_{j}^{\rm AT} + n_{j}^{\rm CO}-1])}$ and $Y_{11}^{(1,[n_{j}^{\rm AT} + n_{j}^{\rm CO}-1])}$ whereas the expectation's difference does not contain these terms, a contradiction, and no unbiased estimator $T$ exists.
\\\\
Second, for $\overline{\rm PE}_{\rm sum}^{\rm AT}(1,\mathbf{k}_{1} - \mathbf{1}; 1,\mathbf{k}_{0} - \mathbf{1};  \mathcal{F})$, consider an unbiased estimator $T(\mathbf{Y}, \mathbf{D}, \mathbf{Z})$ for it and the same two set of parameters $\mathcal{F}$ and $\mathcal{F}'$ above. Then, the difference between the estimands is 
\begin{align*}
&\overline{\rm PE}_{\rm sum}^{\rm AT}(1,\mathbf{k}_{1} - \mathbf{1}; 1,\mathbf{k}_{0} - \mathbf{1};  \mathcal{F}') - \overline{\rm PE}_{\rm sum}^{\rm AT}(1,\mathbf{k}_{1} - \mathbf{1}; 1,\mathbf{k}_{0} - \mathbf{1};  \mathcal{F})\\
=&\left(Y_{11}^{(1,[n_{j}^{\rm AT} + n_{j}^{\rm CO}-1])} - Y_{11}^{(1,[n_{j}^{\rm AT} - 1])}  \right) - \left(Y_{12}^{(1,[n_{j}^{\rm AT} + n_{j}^{\rm CO}-1])} - Y_{12}^{(1,[n_{j}^{\rm AT} -1])} \right)
\end{align*}
The difference in the expectations $E[T(\mathbf{Y}, \mathbf{D}, \mathbf{Z}) \mid \mathcal{F}', \mathcal{Z}]$ and $E[T(\mathbf{Y}, \mathbf{D}, \mathbf{Z}) \mid \mathcal{F}, \mathcal{Z}]$ remains the same as before. But, like the total effect, the terms inside the difference in the estimands, $Y_{12}^{(1,[n_{j}^{\rm AT} + n_{j}^{\rm CO}-1])}$ and $Y_{11}^{(1,[n_{j}^{\rm AT} + n_{j}^{\rm CO}-1])}$, do not appear in the difference of the expectations, leading to a contradiction.
\\\\
Finally, for $\overline{\rm PE}_{\rm sum}^{\rm NT}(0,\mathbf{k}_1; 0, \mathbf{k}_0; \mathcal{F})$, consider an unbiased estimator $T(\mathbf{Y}, \mathbf{D}, \mathbf{Z})$ for it and the same two set of parameters $\mathcal{F}$ and $\mathcal{F}'$ above except we replace the always-takers as never-takers. Then, the difference between the estimands is 
\begin{align*}
&\overline{\rm PE}_{\rm sum}^{\rm NT}(0,\mathbf{k}_1; 0, \mathbf{k}_0; \mathcal{F}') - \overline{\rm PE}_{\rm sum}^{\rm NT}(0,\mathbf{k}_1; 0, \mathbf{k}_0; \mathcal{F}) \\
=&\left(Y_{11}^{(0,[n_{j}^{\rm AT} + n_{j}^{\rm CO}])} - Y_{11}^{(0,[n_{j}^{\rm AT}])}  \right) - \left(Y_{12}^{(0,[n_{j}^{\rm AT} + n_{j}^{\rm CO}])} - Y_{12}^{(0,[n_{j}^{\rm AT}])} \right)
\end{align*}
The difference in the expectations $E[T(\mathbf{Y}, \mathbf{D}, \mathbf{Z}) \mid \mathcal{F}', \mathcal{Z}]$ and $E[T(\mathbf{Y}, \mathbf{D}, \mathbf{Z}) \mid \mathcal{F}, \mathcal{Z}]$ remains the same as before. Like the spillover effect, the terms inside the difference in the estimands, $Y_{12}^{(0,[n_{j}^{\rm AT} + n_{j}^{\rm CO}])}$ and $Y_{11}^{(0,[n_{j}^{\rm AT} + n_{j}^{\rm CO}])}$, do not appear in the difference of the expectations, leading to a contradiction.
\end{proof}

\begin{proof}[Proof of Theorem \ref{thm:bound_est}]
The proof proceeds in two steps. We first show that the estimators $\widehat{\tau}$ and $\frac{N^{\rm CO}}{N^{\rm NT}}$ are consistent. Second, we utilize the continuous mapping theorem to show that the estimators for the bounds are consistent. 

First, under (A1), the numerators and denominators that make up $\widehat{\tau}$, $\widehat{\tau}_Y$ and $\widehat{\tau}_D$ are unbiased for $\tau_Y$ and $\tau_D$, respectively. Let $\widetilde{Y}_{j.}^{(1)} = \sum_{i=1}^{n_j} Y_{ji}^{(1,D_{ji}^{(1)}, \mathbf{D}_{j-i}^{(1)})}$ and $\widetilde{Y}_{j.}^{(0)} = \sum_{i=1}^{n_j} Y_{ji}^{(0,D_{ji}^{(0)}, \mathbf{D}_{j-i}^{(0)})}$. Then, the variances of $\widehat{\tau}_Y$ and $\widehat{\tau}_D$ are
\begin{align*}
Var\left[\widehat{\tau}_Y \mid \mathcal{F}, \mathcal{Z}\right] &= Var\left[\frac{J}{N} \sum_{j=1}^J Z_j \left(\frac{\widetilde{Y}_{j.}^{(1)}}{m} +  \frac{\widetilde{Y}_{j.}^{(0)}}{J-m} \right)  \mid \mathcal{F}, \mathcal{Z} \right] \\
&=  \frac{J^2}{N^2} \cdot \frac{m(J - m)}{J(J-1)} \sum_{j=1}^{J} \left[ \left(\frac{\widetilde{Y}_{j.}^{(1)}}{m} +  \frac{\widetilde{Y}_{j.}^{(0)}}{J-m} \right) - \frac{1}{J} \sum_{j=1}^{J} \left(\frac{\widetilde{Y}_{j.}^{(1)}}{m} +  \frac{\widetilde{Y}_{j.}^{(0)}}{J-m} \right) \right]^2 \\
&\leq \frac{J^2}{N^2} \cdot \frac{m(J - m)}{J(J-1)}  \sum_{j=1}^{J} \left(\frac{\widetilde{Y}_{j.}^{(1)}}{m} +  \frac{\widetilde{Y}_{j.}^{(0)}}{J-m} \right)^2 \\
&\leq J^2 \cdot \frac{m(J - m)}{J(J-1)} \cdot \frac{J^2}{m^2 (J-m)^2} \cdot \frac{ \sum_{j=1}^{J} n_{j}^2 }{\left(\sum_{j=1}^{J} n_j \right)^2} \\
&\leq \frac{J^3}{(J-1)(m)(J-m)} \cdot \frac{JB^2}{J^2} 
\end{align*}
where the second inequality uses the boundedness of the outcome to bound $\widetilde{Y}_{j.}^{(1)}, \widetilde{Y}_{j.}^{(0)} \leq n_j$. As $J, J-m \to \infty$, the upper bound on the variance goes to zero and we have consistency of $\widehat{\tau}_Y$. A similar argument can be used to also prove that $Var\left[\widehat{\tau}_D \mid \mathcal{F}, \mathcal{Z}\right] \to 0$. Combining the two consistent estimators via Slutsky's theorem yields $\widehat{\tau} = \widehat{\tau}_{Y} / \widehat{\tau}_D \to \tau_Y / \tau_D = \tau$.

Next, we show that the ratio $\frac{\widehat{N}^{\rm CO}}{\widehat{N}^{\rm NT}}$ is consistent for the ratio $\frac{p^{\rm CO}}{p^{\rm NT}}$. We see that
\begin{align*}
\frac{\widehat{N}^{\rm CO}}{\widehat{N}^{\rm NT}} &= \frac{\frac{1}{N} \sum_{j=1}^{J} \frac{J}{m} Z_j \sum_{i=1}^{n_j} D_{ji}}{1 - \frac{1}{N} \sum_{j=1}^{J} \frac{J}{m} Z_j \sum_{i=1}^{n_j} D_{ji}}= \frac{\frac{1}{N} \sum_{j=1}^{J} \frac{J}{m} Z_j n_{j}^{\rm CO} }{1 - \frac{1}{N} \sum_{j=1}^{J} \frac{J}{m} Z_j n_{j}^{\rm CO} } 
\end{align*}
Notice that the expectation for the numerator and the denominators are unbiased estimators of the proportions $p^{\rm CO}$ and $p^{\rm NT}$.
\begin{align*}
E\left[\frac{1}{N} \sum_{j=1}^{J} \frac{J}{m} Z_j n_{j}^{\rm CO} \mid \mathcal{F}, \mathcal{Z} \right] &= \frac{N^{\rm CO}}{N} = p^{\rm CO} \\
E\left[1 - \frac{1}{N} \sum_{j=1}^{J} \frac{J}{m} Z_j n_{j}^{\rm CO} \mid \mathcal{F}, \mathcal{Z} \right] &=1 -  \frac{N^{\rm CO}}{N} = p^{\rm NT}
\end{align*} 
Also, the variances of both the numerator and the denominator go to zero:
\begin{align*}
Var\left[\frac{1}{N} \sum_{j=1}^{J} \frac{J}{m} Z_j n_{j}^{\rm CO} \mid \mathcal{F}, \mathcal{Z} \right] &= \frac{J^2}{mN^2} \cdot \frac{J - m}{J(J-1)} \cdot \sum_{j=1}^{J} \left(n_j^{\rm CO} - \frac{1}{J} \sum_{j=1}^{J} n_j^{\rm CO} \right)^2 \\
&\leq \frac{J(J-m)}{m(J-1)} \cdot \frac{\sum_{j=1}^{J} \left(n_j^{\rm CO}\right)^2 }{\left(\sum_{j=1}^{J} n_j \right)^2} \\
&\leq \frac{J(J-m)}{m(J-1)} \cdot \frac{JB^2}{J^2} \to 0
\end{align*}
Consequently, we have
\[
\frac{\widehat{N}^{\rm CO}}{\widehat{N}^{\rm NT}}  \to \frac{p^{\rm CO}}{p^{\rm NT}}
\]
Finally, for any fixed constants $c_0, c_1$, consider the two functions $f(x) = \min(x,c_0)$ for some $c$ and $g(x) = \max(x,c_1)$. These are continuous functions of $x$ and our bounds are of this type, we can use the continuous mapping theorem to arrive at the desired consistency for the bound estimators. 

\end{proof}

\begin{proof}[Proof of Theorem \ref{thm:bound_CI}]
The proof is a direct consequence of Proposition 4 of \citet{kangkeele2018}. Specifically, $H_0$ implies $H_0': \tau = \tau_{x_0,y_0}$ which is the setting in Proposition 4. Also, condition (i) in Proposition 4 of \citet{kangkeele2018} is satisfied because $\mu_Y$ and $\mu_D = p_{\rm CO}$ remain fixed in such a way that $\mu_Y / \mu_D = y_0 + p^{\rm NT}/p^{\rm CO} x_0 = \tau_{x_0,y_0}$ remains fixed; note that $p^{\rm NT} = 1- p^{\rm CO}$. Condition (ii) in Proposition 4 of \citet{kangkeele2018} is satisfied because the outcomes are binary. Then, by Proposition 4 of \citet{kangkeele2018}, we have
\[
T(\tau_{x_0,y_0}) \to N(0,1)
\]
For the plug-in test statistic $T(\widehat{\tau}_{x_0,y_0})$, we have
\begin{align*}
T(\widehat{\tau}_{x_0,y_0}) &= \frac{ \frac{N}{J} \left( \widehat{\tau}_{Y} - \widehat{\tau}_{x_0,y_0} \widehat{\tau}_D\right) }{ \sqrt{ {\rm Var} \left[ \frac{N}{J} \left( \widehat{\tau}_{Y} - \widehat{\tau}_{x_0,y_0} \widehat{\tau}_D\right) \mid \mathcal{F}, \mathcal{Z} \right]}} \\
&=\sqrt{ \frac{ {\rm Var} \left[\widehat{\tau}_Y - {\tau}_{x_0,y_0} \hat{\tau}_D  \mid \mathcal{F}, \mathcal{Z} \right] }{ {\rm Var} \left[\widehat{\tau}_Y - \widehat{\tau}_{x_0,y_0} \hat{\tau}_D  \mid \mathcal{F}, \mathcal{Z} \right]}} \left( \frac{ \widehat{\tau}_Y - \tau_{x_0,y_0}\hat{\tau}_D}{\sqrt{ {\rm Var} \left[\widehat{\tau}_Y - {\tau}_{x_0,y_0} \hat{\tau}_D  \mid \mathcal{F}, \mathcal{Z} \right]}}  + \frac{(\tau_{x_0,y_0} - \widehat{\tau}_{x_0,y_0} ) \widehat{\tau}_{D} }{\sqrt{ {\rm Var} \left[\widehat{\tau}_Y - {\tau}_{x_0,y_0} \hat{\tau}_D  \mid \mathcal{F}, \mathcal{Z} \right]}}\right) \\
&=  \sqrt{ \frac{ {\rm Var} \left[\widehat{\tau}_Y - {\tau}_{x_0,y_0} \hat{\tau}_D  \mid \mathcal{F}, \mathcal{Z} \right] }{ {\rm Var} \left[\widehat{\tau}_Y - \widehat{\tau}_{x_0,y_0} \hat{\tau}_D  \mid \mathcal{F}, \mathcal{Z} \right]}} T(\tau_{x_0,y_0}) \\
&\quad{} +  \sqrt{ \frac{ {\rm Var} \left[\widehat{\tau}_Y - {\tau}_{x_0,y_0} \hat{\tau}_D  \mid \mathcal{F}, \mathcal{Z} \right] }{ {\rm Var} \left[\widehat{\tau}_Y - \widehat{\tau}_{x_0,y_0} \hat{\tau}_D  \mid \mathcal{F}, \mathcal{Z} \right]}} \frac{(\tau_{x_0,y_0} - \widehat{\tau}_{x_0,y_0} ) \widehat{\tau}_{D} }{\sqrt{ {\rm Var} \left[\widehat{\tau}_Y - {\tau}_{x_0,y_0} \hat{\tau}_D  \mid \mathcal{F}, \mathcal{Z} \right]}}
\end{align*}
By the proof in Theorem \ref{thm:bound_est} and under $H_0$, the following convergences in probability are true
\begin{align*}
\widehat{\tau}_{D} = \widehat{p}^{\rm CO} \to \tau_D = p^{\rm CO}, \quad{} \widehat{\tau}_Y \to \tau_Y, \quad{} \widehat{\tau}_{x_0,y_0} \to \tau_{x_0,y_0}
\end{align*}
so that the product converges to zero in probability, i.e. $(\tau_{x_0,y_0} - \widehat{\tau}_{x_0,y_0}) \widehat{\tau}_D \to 0$. We also note that because $Y_{ji}, D_{ji}$ are binary and the cluster-sized are bounded, both $\widehat{\tau}_Y$ and $\widehat{\tau}_D$ are uniformly bounded. Then, by the dominated convergence theorem, under $H_0$, we have
\begin{align*}
E[ (\widehat{\tau}_Y - \tau_{x_0,y_0} \widehat{\tau}_D)^2 \mid \mathcal{F}, \mathcal{Z}] &\to (\tau_Y - \tau_{x_0,y_0} \tau_{D})^2, \quad{} E[ (\widehat{\tau}_Y - \tau_{x_0,y_0} \widehat{\tau}_D) \mid \mathcal{F}, \mathcal{Z}] \to (\tau_Y - \tau_{x_0,y_0} \tau_{D}) \\
E[ (\widehat{\tau}_Y - \widehat{\tau}_{x_0,y_0} \widehat{\tau}_D)^2 \mid \mathcal{F}, \mathcal{Z}] &\to (\tau_Y - \tau_{x_0,y_0} \tau_{D})^2, \quad{} E[ (\widehat{\tau}_Y - \widehat{\tau}_{x_0,y_0} \widehat{\tau}_D) \mid \mathcal{F}, \mathcal{Z}] \to (\tau_Y - \tau_{x_0,y_0} \tau_{D}) 
\end{align*}
Hence, 
\[
 \sqrt{ \frac{ {\rm Var} \left[\widehat{\tau}_Y - {\tau}_{x_0,y_0} \hat{\tau}_D  \mid \mathcal{F}, \mathcal{Z} \right] }{ {\rm Var} \left[\widehat{\tau}_Y - \widehat{\tau}_{x_0,y_0} \hat{\tau}_D  \mid \mathcal{F}, \mathcal{Z} \right]}} =   \sqrt{ \frac{ {\rm E} \left[ \left(\widehat{\tau}_Y - {\tau}_{x_0,y_0} \hat{\tau}_D \right)^2  \mid \mathcal{F}, \mathcal{Z} \right] - {\rm E} \left[ \left(\widehat{\tau}_Y - {\tau}_{x_0,y_0} \hat{\tau}_D \right)  \mid \mathcal{F}, \mathcal{Z} \right]^2}{ {\rm E} \left[ \left(\widehat{\tau}_Y - \widehat{\tau}_{x_0,y_0} \widehat{\tau}_D \right)^2  \mid \mathcal{F}, \mathcal{Z} \right] - {\rm E} \left[ \left(\widehat{\tau}_Y - \widehat{\tau}_{x_0,y_0} \widehat{\tau}_D \right)  \mid \mathcal{F}, \mathcal{Z} \right]^2}} \to 1
\]
and by Slutsky's theorem, we arrive at the desired result.
\end{proof}

\section*{Acknowledgements}
We thank Guillaume Basse, Guanglei Hong and participants at the University of Chicago Quantitative Methods Committee in the Social Sciences, Berkeley Research Workshop in Quantitative Modeling, University of Wisconsin-Madison Demography Seminar, 2018 NetSci Conference, 2018 Atlantic Causal Inference Conference, and Penn Causal Inference Seminar for helpful comments.

\begin{supplement}
\label{suppA}
\stitle{Replication Materials for Empirical Analysis}
\slink[url]{http://www.e-publications.org/ims/support/download/xxx}
\sdescription{The online supplement contains code and replication materials for the }
\end{supplement}

\clearpage
\bibliographystyle{imsart-nameyear}
\bibliography{ref}

\begin{thebibliography}{59}

\bibitem[\protect\citeauthoryear{Aiken et~al.}{2015}]{aiken2015re}
\begin{barticle}[author]
\bauthor{\bsnm{Aiken},~\bfnm{Alexander~M}\binits{A.~M.}},
  \bauthor{\bsnm{Davey},~\bfnm{Calum}\binits{C.}},
  \bauthor{\bsnm{Hargreaves},~\bfnm{James~R}\binits{J.~R.}} \AND
  \bauthor{\bsnm{Hayes},~\bfnm{Richard~J}\binits{R.~J.}}
(\byear{2015}).
\btitle{Re-analysis of health and educational impacts of a school-based
  deworming programme in western Kenya: a pure replication}.
\bjournal{International Journal of Epidemiology}
\bvolume{44}
\bpages{1572--1580}.
\end{barticle}
\endbibitem

\bibitem[\protect\citeauthoryear{Andrews}{2000}]{andrews2000inconsistency}
\begin{barticle}[author]
\bauthor{\bsnm{Andrews},~\bfnm{Donald~WK}\binits{D.~W.}}
(\byear{2000}).
\btitle{Inconsistency of the bootstrap when a parameter is on the boundary of
  the parameter space}.
\bjournal{Econometrica}
\bvolume{68}
\bpages{399--405}.
\end{barticle}
\endbibitem

\bibitem[\protect\citeauthoryear{Andrews and
  Guggenberger}{2009}]{andrews2009validity}
\begin{barticle}[author]
\bauthor{\bsnm{Andrews},~\bfnm{Donald~WK}\binits{D.~W.}} \AND
  \bauthor{\bsnm{Guggenberger},~\bfnm{Patrik}\binits{P.}}
(\byear{2009}).
\btitle{Validity of subsampling and ``plug-in asymptotic'' inference for
  parameters defined by moment inequalities}.
\bjournal{Econometric Theory}
\bvolume{25}
\bpages{669--709}.
\end{barticle}
\endbibitem

\bibitem[\protect\citeauthoryear{Angrist, Imbens and
  Rubin}{1996}]{angrist_identification_1996}
\begin{barticle}[author]
\bauthor{\bsnm{Angrist},~\bfnm{Joshua~D.}\binits{J.~D.}},
  \bauthor{\bsnm{Imbens},~\bfnm{Guido~W.}\binits{G.~W.}} \AND
  \bauthor{\bsnm{Rubin},~\bfnm{Donald~B.}\binits{D.~B.}}
(\byear{1996}).
\btitle{Identification of Causal Effects Using Instrumental Variables}.
\bjournal{Journal of the American Statistical Association}
\bvolume{91}
\bpages{444--455}.
\end{barticle}
\endbibitem

\bibitem[\protect\citeauthoryear{Angrist and
  Krueger}{2001}]{angrist_instrumental_2001}
\begin{barticle}[author]
\bauthor{\bsnm{Angrist},~\bfnm{Joshua~D}\binits{J.~D.}} \AND
  \bauthor{\bsnm{Krueger},~\bfnm{Alan~B}\binits{A.~B.}}
(\byear{2001}).
\btitle{Instrumental variables and the search for identification: From supply
  and demand to natural experiments}.
\bjournal{Journal of Economic perspectives}
\bvolume{15}
\bpages{69--85}.
\end{barticle}
\endbibitem

\bibitem[\protect\citeauthoryear{Aronow and Samii}{2017}]{AronowSamii2017-AOAS}
\begin{barticle}[author]
\bauthor{\bsnm{Aronow},~\bfnm{Peter~M.}\binits{P.~M.}} \AND
  \bauthor{\bsnm{Samii},~\bfnm{Cyrus}\binits{C.}}
(\byear{2017}).
\btitle{Estimating Average Causal Effects Under Interference Between Units}.
\bjournal{\emph{forthcoming}, Annals of Applied Statistics}.
\end{barticle}
\endbibitem

\bibitem[\protect\citeauthoryear{Baiocchi, Cheng and
  Small}{2014}]{baiocchi_instrumental_2014}
\begin{barticle}[author]
\bauthor{\bsnm{Baiocchi},~\bfnm{Michael}\binits{M.}},
  \bauthor{\bsnm{Cheng},~\bfnm{Jing}\binits{J.}} \AND
  \bauthor{\bsnm{Small},~\bfnm{Dylan~S}\binits{D.~S.}}
(\byear{2014}).
\btitle{Instrumental variable methods for causal inference}.
\bjournal{Statistics in Medicine}
\bvolume{33}
\bpages{2297--2340}.
\end{barticle}
\endbibitem

\bibitem[\protect\citeauthoryear{Baiocchi
  et~al.}{2010}]{baiocchi_building_2010}
\begin{barticle}[author]
\bauthor{\bsnm{Baiocchi},~\bfnm{Mike}\binits{M.}},
  \bauthor{\bsnm{Small},~\bfnm{Dylan~S.}\binits{D.~S.}},
  \bauthor{\bsnm{Lorch},~\bfnm{Scott}\binits{S.}} \AND
  \bauthor{\bsnm{Rosenbaum},~\bfnm{Paul~R.}\binits{P.~R.}}
(\byear{2010}).
\btitle{Building a Stronger Instrument in an Observational Study of Perinatal
  Care for Premature Infants}.
\bjournal{Journal of the American Statistical Association}
\bvolume{105}
\bpages{1285-1296}.
\end{barticle}
\endbibitem

\bibitem[\protect\citeauthoryear{Basse and
  Airoldi}{2018}]{basse2018limitations}
\begin{barticle}[author]
\bauthor{\bsnm{Basse},~\bfnm{Guillaume~W}\binits{G.~W.}} \AND
  \bauthor{\bsnm{Airoldi},~\bfnm{Edoardo~M}\binits{E.~M.}}
(\byear{2018}).
\btitle{Limitations of design-based causal inference and A/B testing under
  arbitrary and network interference}.
\bjournal{Sociological Methodology}
\bvolume{48}
\bpages{136--151}.
\end{barticle}
\endbibitem

\bibitem[\protect\citeauthoryear{Bowers, Fredrickson and
  Panagopoulos}{2013}]{BowersFredricksonPanagopoulos2013}
\begin{barticle}[author]
\bauthor{\bsnm{Bowers},~\bfnm{Jake}\binits{J.}},
  \bauthor{\bsnm{Fredrickson},~\bfnm{Mark~M}\binits{M.~M.}} \AND
  \bauthor{\bsnm{Panagopoulos},~\bfnm{Costas}\binits{C.}}
(\byear{2013}).
\btitle{Reasoning about Interference Between Units: A General Framework}.
\bjournal{Political Analysis}
\bvolume{21}
\bpages{97--124}.
\end{barticle}
\endbibitem

\bibitem[\protect\citeauthoryear{Cheng and Small}{2006}]{cheng_bounds_2006}
\begin{barticle}[author]
\bauthor{\bsnm{Cheng},~\bfnm{Jing}\binits{J.}} \AND
  \bauthor{\bsnm{Small},~\bfnm{Dylan~S}\binits{D.~S.}}
(\byear{2006}).
\btitle{Bounds on causal effects in three-arm trials with non-compliance}.
\bjournal{Journal of the Royal Statistical Society: Series B (Statistical
  Methodology)}
\bvolume{68}
\bpages{815--836}.
\end{barticle}
\endbibitem

\bibitem[\protect\citeauthoryear{Chernozhukov, Lee and
  Rosen}{2013}]{chernozhukov2013intersection}
\begin{barticle}[author]
\bauthor{\bsnm{Chernozhukov},~\bfnm{Victor}\binits{V.}},
  \bauthor{\bsnm{Lee},~\bfnm{Sokbae}\binits{S.}} \AND
  \bauthor{\bsnm{Rosen},~\bfnm{Adam~M}\binits{A.~M.}}
(\byear{2013}).
\btitle{Intersection bounds: estimation and inference}.
\bjournal{Econometrica}
\bvolume{81}
\bpages{667--737}.
\end{barticle}
\endbibitem

\bibitem[\protect\citeauthoryear{Choi}{2017}]{choi_estimation_2017}
\begin{barticle}[author]
\bauthor{\bsnm{Choi},~\bfnm{David}\binits{D.}}
(\byear{2017}).
\btitle{Estimation of monotone treatment effects in network experiments}.
\bjournal{Journal of the American Statistical Association}
\bvolume{112}
\bpages{1147--1155}.
\end{barticle}
\endbibitem

\bibitem[\protect\citeauthoryear{Cook and Campbell}{1979}]{cook_quasi_1979}
\begin{bbook}[author]
\bauthor{\bsnm{Cook},~\bfnm{Thomas~D}\binits{T.~D.}} \AND
  \bauthor{\bsnm{Campbell},~\bfnm{Donald~Thomas}\binits{D.~T.}}
(\byear{1979}).
\btitle{Quasi-experimentation: Design and analysis for field settings}
\bvolume{3}.
\bpublisher{Rand McNally Chicago}.
\end{bbook}
\endbibitem

\bibitem[\protect\citeauthoryear{Cox}{1958}]{cox_planning_1958}
\begin{bbook}[author]
\bauthor{\bsnm{Cox},~\bfnm{D.~R.}\binits{D.~R.}}
(\byear{1958}).
\btitle{Planning of Experiments}.
\bpublisher{Wiley}, \baddress{New York, NY}.
\end{bbook}
\endbibitem

\bibitem[\protect\citeauthoryear{Deaton}{2010}]{Deaton:2010}
\begin{barticle}[author]
\bauthor{\bsnm{Deaton},~\bfnm{Angus}\binits{A.}}
(\byear{2010}).
\btitle{Instruments, randomization, and learning about development}.
\bjournal{Journal of Economic Literature}
\bpages{424--455}.
\end{barticle}
\endbibitem

\bibitem[\protect\citeauthoryear{Forastiere, Mealli and
  VanderWeele}{2016}]{forastiere2016identification}
\begin{barticle}[author]
\bauthor{\bsnm{Forastiere},~\bfnm{Laura}\binits{L.}},
  \bauthor{\bsnm{Mealli},~\bfnm{Fabrizia}\binits{F.}} \AND
  \bauthor{\bsnm{VanderWeele},~\bfnm{Tyler~J}\binits{T.~J.}}
(\byear{2016}).
\btitle{Identification and estimation of causal mechanisms in clustered
  encouragement designs: Disentangling bed nets using Bayesian principal
  stratification}.
\bjournal{Journal of the American Statistical Association}
\bvolume{111}
\bpages{510--525}.
\end{barticle}
\endbibitem

\bibitem[\protect\citeauthoryear{Frangakis, Rubin and
  Zhou}{2002}]{frangakis_clustered_2002}
\begin{barticle}[author]
\bauthor{\bsnm{Frangakis},~\bfnm{Constantine~E}\binits{C.~E.}},
  \bauthor{\bsnm{Rubin},~\bfnm{Donald~B}\binits{D.~B.}} \AND
  \bauthor{\bsnm{Zhou},~\bfnm{Xiao-Hua}\binits{X.-H.}}
(\byear{2002}).
\btitle{Clustered encouragement designs with individual noncompliance: Bayesian
  inference with randomization, and application to advance directive forms}.
\bjournal{Biostatistics}
\bvolume{3}
\bpages{147--164}.
\end{barticle}
\endbibitem

\bibitem[\protect\citeauthoryear{Halloran and
  Struchiner}{1991}]{halloran_study_1991}
\begin{barticle}[author]
\bauthor{\bsnm{Halloran},~\bfnm{M.~Elizabeth}\binits{M.~E.}} \AND
  \bauthor{\bsnm{Struchiner},~\bfnm{Claudio~J.}\binits{C.~J.}}
(\byear{1991}).
\btitle{Study Designs for Dependent Happenings}.
\bjournal{Epidemiology}
\bvolume{2}
\bpages{331--338}.
\end{barticle}
\endbibitem

\bibitem[\protect\citeauthoryear{Halloran and
  Struchiner}{1995}]{halloran_causal_1995}
\begin{barticle}[author]
\bauthor{\bsnm{Halloran},~\bfnm{M.~Elizabeth}\binits{M.~E.}} \AND
  \bauthor{\bsnm{Struchiner},~\bfnm{Claudio~J.}\binits{C.~J.}}
(\byear{1995}).
\btitle{Causal Inference in Infectious Diseases}.
\bjournal{Epidemiology}
\bvolume{6}
\bpages{142--151}.
\end{barticle}
\endbibitem

\bibitem[\protect\citeauthoryear{Hern{\'a}n and
  Robins}{2006}]{hernan_instruments_2006}
\begin{barticle}[author]
\bauthor{\bsnm{Hern{\'a}n},~\bfnm{Miguel~A}\binits{M.~A.}} \AND
  \bauthor{\bsnm{Robins},~\bfnm{James~M}\binits{J.~M.}}
(\byear{2006}).
\btitle{Instruments for causal inference: an epidemiologist's dream?}
\bjournal{Epidemiology}
\bvolume{17}
\bpages{360--372}.
\end{barticle}
\endbibitem

\bibitem[\protect\citeauthoryear{Hicks, Kremer and
  Miguel}{2015}]{hicks2015commentary}
\begin{barticle}[author]
\bauthor{\bsnm{Hicks},~\bfnm{Joan~Hamory}\binits{J.~H.}},
  \bauthor{\bsnm{Kremer},~\bfnm{Michael}\binits{M.}} \AND
  \bauthor{\bsnm{Miguel},~\bfnm{Edward}\binits{E.}}
(\byear{2015}).
\btitle{Commentary: Deworming externalities and schooling impacts in Kenya: a
  comment on Aiken et al.(2015) and Davey et al.(2015)}.
\bjournal{International Journal of Epidemiology}
\bvolume{44}
\bpages{1593--1596}.
\end{barticle}
\endbibitem

\bibitem[\protect\citeauthoryear{Hirano and
  Porter}{2012}]{hirano2012impossibility}
\begin{barticle}[author]
\bauthor{\bsnm{Hirano},~\bfnm{Keisuke}\binits{K.}} \AND
  \bauthor{\bsnm{Porter},~\bfnm{Jack~R}\binits{J.~R.}}
(\byear{2012}).
\btitle{Impossibility results for nondifferentiable functionals}.
\bjournal{Econometrica}
\bvolume{80}
\bpages{1769--1790}.
\end{barticle}
\endbibitem

\bibitem[\protect\citeauthoryear{Hong}{2015}]{hong2015causality}
\begin{bbook}[author]
\bauthor{\bsnm{Hong},~\bfnm{Guanglei}\binits{G.}}
(\byear{2015}).
\btitle{Causality in a social world: Moderation, mediation and spill-over}.
\bpublisher{John Wiley \& Sons}.
\end{bbook}
\endbibitem

\bibitem[\protect\citeauthoryear{Hong and
  Raudenbush}{2006}]{hong2006evaluating}
\begin{barticle}[author]
\bauthor{\bsnm{Hong},~\bfnm{Guanglei}\binits{G.}} \AND
  \bauthor{\bsnm{Raudenbush},~\bfnm{Stephen~W}\binits{S.~W.}}
(\byear{2006}).
\btitle{Evaluating kindergarten retention policy: A case study of causal
  inference for multilevel observational data}.
\bjournal{Journal of the American Statistical Association}
\bvolume{101}
\bpages{901--910}.
\end{barticle}
\endbibitem

\bibitem[\protect\citeauthoryear{Horowitz and
  Manski}{2000}]{horowitz_nonparametric_2000}
\begin{barticle}[author]
\bauthor{\bsnm{Horowitz},~\bfnm{Joel~L}\binits{J.~L.}} \AND
  \bauthor{\bsnm{Manski},~\bfnm{Charles~F}\binits{C.~F.}}
(\byear{2000}).
\btitle{Nonparametric analysis of randomized experiments with missing covariate
  and outcome data}.
\bjournal{Journal of the American statistical Association}
\bvolume{95}
\bpages{77--84}.
\end{barticle}
\endbibitem

\bibitem[\protect\citeauthoryear{Hudgens and
  Halloran}{2008}]{hudgens_toward_2008}
\begin{barticle}[author]
\bauthor{\bsnm{Hudgens},~\bfnm{Michael~G}\binits{M.~G.}} \AND
  \bauthor{\bsnm{Halloran},~\bfnm{M~Elizabeth}\binits{M.~E.}}
(\byear{2008}).
\btitle{Toward causal inference with interference}.
\bjournal{Journal of the American Statistical Association}
\bvolume{103}
\bpages{832--842}.
\end{barticle}
\endbibitem

\bibitem[\protect\citeauthoryear{Imai, Jiang and Malai}{2018}]{imai2018causal}
\begin{bunpublished}[author]
\bauthor{\bsnm{Imai},~\bfnm{Kosuke}\binits{K.}},
  \bauthor{\bsnm{Jiang},~\bfnm{Zhichao}\binits{Z.}} \AND
  \bauthor{\bsnm{Malai},~\bfnm{Anup}\binits{A.}}
(\byear{2018}).
\btitle{Causal Inference with Interference and Noncompliance in Two-Stage
  Randomized Experiments}.
\bnote{Unpublished Manuscript}.
\end{bunpublished}
\endbibitem

\bibitem[\protect\citeauthoryear{Imai et~al.}{2009}]{imai_essential_2009}
\begin{barticle}[author]
\bauthor{\bsnm{Imai},~\bfnm{Kosuke}\binits{K.}},
  \bauthor{\bsnm{King},~\bfnm{Gary}\binits{G.}},
  \bauthor{\bsnm{Nall},~\bfnm{Clayton}\binits{C.}} \betal{et~al.}
(\byear{2009}).
\btitle{The essential role of pair matching in cluster-randomized experiments,
  with application to the Mexican universal health insurance evaluation}.
\bjournal{Statistical Science}
\bvolume{24}
\bpages{29--53}.
\end{barticle}
\endbibitem

\bibitem[\protect\citeauthoryear{Imbens}{2010}]{Imbens:2010}
\begin{barticle}[author]
\bauthor{\bsnm{Imbens},~\bfnm{Guido~W.}\binits{G.~W.}}
(\byear{2010}).
\btitle{Better LATE Than Nothing: Some Comments on Deaton (2009) and Heckman
  and Urzua (2009)}.
\bjournal{Journal of Economic Literature}
\bvolume{48}
\bpages{399-423}.
\end{barticle}
\endbibitem

\bibitem[\protect\citeauthoryear{Imbens}{2014}]{imbens_instrumental_2014}
\begin{barticle}[author]
\bauthor{\bsnm{Imbens},~\bfnm{Guido~W}\binits{G.~W.}}
(\byear{2014}).
\btitle{Instrumental Variables: An Econometrician's Perspective}.
\bjournal{Statistical Science}
\bpages{323--358}.
\end{barticle}
\endbibitem

\bibitem[\protect\citeauthoryear{Imbens and
  Angrist}{1994}]{imbens_identification_1994}
\begin{barticle}[author]
\bauthor{\bsnm{Imbens},~\bfnm{Guido~W.}\binits{G.~W.}} \AND
  \bauthor{\bsnm{Angrist},~\bfnm{Joshua~D.}\binits{J.~D.}}
(\byear{1994}).
\btitle{Identification and Estimation of Local Average Treatment Effects}.
\bjournal{Econometrica}
\bvolume{62}
\bpages{467--475}.
\end{barticle}
\endbibitem

\bibitem[\protect\citeauthoryear{Imbens and
  Manski}{2004}]{imbens_confidence_2004}
\begin{barticle}[author]
\bauthor{\bsnm{Imbens},~\bfnm{Guido~W}\binits{G.~W.}} \AND
  \bauthor{\bsnm{Manski},~\bfnm{Charles~F}\binits{C.~F.}}
(\byear{2004}).
\btitle{Confidence intervals for partially identified parameters}.
\bjournal{Econometrica}
\bvolume{72}
\bpages{1845--1857}.
\end{barticle}
\endbibitem

\bibitem[\protect\citeauthoryear{Imbens and
  Wooldridge}{2008}]{imbens2008recent}
\begin{barticle}[author]
\bauthor{\bsnm{Imbens},~\bfnm{Guido~M}\binits{G.~M.}} \AND
  \bauthor{\bsnm{Wooldridge},~\bfnm{Jeffrey~M}\binits{J.~M.}}
(\byear{2008}).
\btitle{Recent developments in the econometrics of program evaluation}.
\bjournal{Journal of Economic Literature}
\bvolume{47}
\bpages{5--86}.
\end{barticle}
\endbibitem

\bibitem[\protect\citeauthoryear{Jo, Asparouhov and
  Muth{\'e}n}{2008}]{jo_intention_2008}
\begin{barticle}[author]
\bauthor{\bsnm{Jo},~\bfnm{Booil}\binits{B.}},
  \bauthor{\bsnm{Asparouhov},~\bfnm{Tihomir}\binits{T.}} \AND
  \bauthor{\bsnm{Muth{\'e}n},~\bfnm{Bengt~O}\binits{B.~O.}}
(\byear{2008}).
\btitle{Intention-to-treat analysis in cluster randomized trials with
  noncompliance}.
\bjournal{Statistics in medicine}
\bvolume{27}
\bpages{5565--5577}.
\end{barticle}
\endbibitem

\bibitem[\protect\citeauthoryear{Kang and Imbens}{2016}]{kang_peer_2016}
\begin{barticle}[author]
\bauthor{\bsnm{Kang},~\bfnm{Hyunseung}\binits{H.}} \AND
  \bauthor{\bsnm{Imbens},~\bfnm{Guido}\binits{G.}}
(\byear{2016}).
\btitle{Peer Encouragement Designs in Causal Inference with Partial
  Interference and Identification of Local Average Network Effects}.
\bjournal{arXiv preprint arXiv:1609.04464}.
\end{barticle}
\endbibitem

\bibitem[\protect\citeauthoryear{Kang and Keele}{2018}]{kangkeele2018}
\begin{bunpublished}[author]
\bauthor{\bsnm{Kang},~\bfnm{Hyunseung}\binits{H.}} \AND
  \bauthor{\bsnm{Keele},~\bfnm{Luke}\binits{L.}}
(\byear{2018}).
\btitle{Estimation Methods for Cluster Randomized Trials with Noncompliance: A
  Study of A Biometric Smartcard Payment System in India}.
\bnote{Unpublished Manuscript}.
\end{bunpublished}
\endbibitem

\bibitem[\protect\citeauthoryear{Lehmann}{2004}]{lehmann_elements_2004}
\begin{bbook}[author]
\bauthor{\bsnm{Lehmann},~\bfnm{Erich~Leo}\binits{E.~L.}}
(\byear{2004}).
\btitle{Elements of Large-Sample Theory}.
\bpublisher{Springer Science \& Business Media}.
\end{bbook}
\endbibitem

\bibitem[\protect\citeauthoryear{Liu and Hudgens}{2014}]{liu_large_2014}
\begin{barticle}[author]
\bauthor{\bsnm{Liu},~\bfnm{Lan}\binits{L.}} \AND
  \bauthor{\bsnm{Hudgens},~\bfnm{Michael~G}\binits{M.~G.}}
(\byear{2014}).
\btitle{Large sample randomization inference of causal effects in the presence
  of interference}.
\bjournal{Journal of the American Statistical Association}
\bvolume{109}
\bpages{288--301}.
\end{barticle}
\endbibitem

\bibitem[\protect\citeauthoryear{Manski}{1993}]{manski1993identification}
\begin{barticle}[author]
\bauthor{\bsnm{Manski},~\bfnm{Charles~F}\binits{C.~F.}}
(\byear{1993}).
\btitle{Identification of endogenous social effects: The reflection problem}.
\bjournal{The Review of Economic Studies}
\bvolume{60}
\bpages{531--542}.
\end{barticle}
\endbibitem

\bibitem[\protect\citeauthoryear{Manski}{2013}]{manski2013identification}
\begin{barticle}[author]
\bauthor{\bsnm{Manski},~\bfnm{Charles~F}\binits{C.~F.}}
(\byear{2013}).
\btitle{Identification of treatment response with social interactions}.
\bjournal{The Econometrics Journal}
\bvolume{16}
\bpages{S1--S23}.
\end{barticle}
\endbibitem

\bibitem[\protect\citeauthoryear{Miguel and Kremer}{2004}]{miguel2004worms}
\begin{barticle}[author]
\bauthor{\bsnm{Miguel},~\bfnm{Edward}\binits{E.}} \AND
  \bauthor{\bsnm{Kremer},~\bfnm{Michael}\binits{M.}}
(\byear{2004}).
\btitle{Worms: identifying impacts on education and health in the presence of
  treatment externalities}.
\bjournal{Econometrica}
\bvolume{72}
\bpages{159--217}.
\end{barticle}
\endbibitem

\bibitem[\protect\citeauthoryear{Murray}{1998}]{murray_design_1998}
\begin{bbook}[author]
\bauthor{\bsnm{Murray},~\bfnm{David~M}\binits{D.~M.}}
(\byear{1998}).
\btitle{Design and analysis of group-randomized trials}
\bvolume{29}.
\bpublisher{Monographs in Epidemiology}.
\end{bbook}
\endbibitem

\bibitem[\protect\citeauthoryear{Nelson and Startz}{1990}]{nelson1990some}
\begin{barticle}[author]
\bauthor{\bsnm{Nelson},~\bfnm{Charles~R}\binits{C.~R.}} \AND
  \bauthor{\bsnm{Startz},~\bfnm{Richard}\binits{R.}}
(\byear{1990}).
\btitle{Some Further Results on the Exact Small Sample Properties of the
  Instrumental Variable Estimator}.
\bjournal{Econometrica}
\bpages{967--976}.
\end{barticle}
\endbibitem

\bibitem[\protect\citeauthoryear{Neyman}{1923}]{neyman_application_1990}
\begin{barticle}[author]
\bauthor{\bsnm{Neyman},~\bfnm{Jerzy}\binits{J.}}
(\byear{1923}).
\btitle{On the Application of Probability Theory to Agricultural Experiments.
  Essay on Principles. Section 9.}
\bjournal{Statistical Science}
\bvolume{5}
\bpages{465-472}.
\end{barticle}
\endbibitem

\bibitem[\protect\citeauthoryear{Romano}{1989}]{romano1989bootstrap}
\begin{barticle}[author]
\bauthor{\bsnm{Romano},~\bfnm{Joseph~P}\binits{J.~P.}}
(\byear{1989}).
\btitle{Do bootstrap confidence procedures behave well uniformly in P?}
\bjournal{Canadian Journal of Statistics}
\bvolume{17}
\bpages{75--80}.
\end{barticle}
\endbibitem

\bibitem[\protect\citeauthoryear{Romano and Shaikh}{2008}]{romano2008inference}
\begin{barticle}[author]
\bauthor{\bsnm{Romano},~\bfnm{Joseph~P}\binits{J.~P.}} \AND
  \bauthor{\bsnm{Shaikh},~\bfnm{Azeem~M}\binits{A.~M.}}
(\byear{2008}).
\btitle{Inference for identifiable parameters in partially identified
  econometric models}.
\bjournal{Journal of Statistical Planning and Inference}
\bvolume{138}
\bpages{2786--2807}.
\end{barticle}
\endbibitem

\bibitem[\protect\citeauthoryear{Romano and Shaikh}{2010}]{romano2010inference}
\begin{barticle}[author]
\bauthor{\bsnm{Romano},~\bfnm{Joseph~P}\binits{J.~P.}} \AND
  \bauthor{\bsnm{Shaikh},~\bfnm{Azeem~M}\binits{A.~M.}}
(\byear{2010}).
\btitle{Inference for the identified set in partially identified econometric
  models}.
\bjournal{Econometrica}
\bvolume{78}
\bpages{169--211}.
\end{barticle}
\endbibitem

\bibitem[\protect\citeauthoryear{Rosenbaum}{2007}]{Rosenbaum2007-JASA}
\begin{barticle}[author]
\bauthor{\bsnm{Rosenbaum},~\bfnm{Paul~R.}\binits{P.~R.}}
(\byear{2007}).
\btitle{Interference Between Units in Randomized Experiments}.
\bjournal{Journal of the American Statistical Association}
\bvolume{102}
\bpages{191-200}.
\end{barticle}
\endbibitem

\bibitem[\protect\citeauthoryear{Rubin}{1974}]{rubin_estimating_1974}
\begin{barticle}[author]
\bauthor{\bsnm{Rubin},~\bfnm{Donald~B.}\binits{D.~B.}}
(\byear{1974}).
\btitle{Estimating causal effects of treatments in randomized and nonrandomized
  studies.}
\bjournal{Journal of Educational Psychology}
\bvolume{66}
\bpages{688}.
\end{barticle}
\endbibitem

\bibitem[\protect\citeauthoryear{Schochet and
  Chiang}{2011}]{schochet_estimation_2011}
\begin{barticle}[author]
\bauthor{\bsnm{Schochet},~\bfnm{Peter~Z}\binits{P.~Z.}} \AND
  \bauthor{\bsnm{Chiang},~\bfnm{Hanley~S}\binits{H.~S.}}
(\byear{2011}).
\btitle{Estimation and identification of the complier average causal effect
  parameter in education RCTs}.
\bjournal{Journal of Educational and Behavioral Statistics}
\bvolume{36}
\bpages{307--345}.
\end{barticle}
\endbibitem

\bibitem[\protect\citeauthoryear{Shalizi and
  Thomas}{2011}]{shalizi2011homophily}
\begin{barticle}[author]
\bauthor{\bsnm{Shalizi},~\bfnm{Cosma~Rohilla}\binits{C.~R.}} \AND
  \bauthor{\bsnm{Thomas},~\bfnm{Andrew~C}\binits{A.~C.}}
(\byear{2011}).
\btitle{Homophily and contagion are generically confounded in observational
  social network studies}.
\bjournal{Sociological Methods \& Research}
\bvolume{40}
\bpages{211--239}.
\end{barticle}
\endbibitem

\bibitem[\protect\citeauthoryear{Small, Ten~Have and
  Rosenbaum}{2008}]{small_randomization_2008}
\begin{barticle}[author]
\bauthor{\bsnm{Small},~\bfnm{Dylan~S}\binits{D.~S.}},
  \bauthor{\bsnm{Ten~Have},~\bfnm{Thomas~R}\binits{T.~R.}} \AND
  \bauthor{\bsnm{Rosenbaum},~\bfnm{Paul~R}\binits{P.~R.}}
(\byear{2008}).
\btitle{Randomization inference in a group--randomized trial of treatments for
  depression: covariate adjustment, noncompliance, and quantile effects}.
\bjournal{Journal of the American Statistical Association}
\bvolume{103}
\bpages{271--279}.
\end{barticle}
\endbibitem

\bibitem[\protect\citeauthoryear{Sobel}{2006}]{Sobel2006-JASA}
\begin{barticle}[author]
\bauthor{\bsnm{Sobel},~\bfnm{Michael~E}\binits{M.~E.}}
(\byear{2006}).
\btitle{What do randomized studies of housing mobility demonstrate? Causal
  inference in the face of interference}.
\bjournal{Journal of the American Statistical Association}
\bvolume{101}
\bpages{1398--1407}.
\end{barticle}
\endbibitem

\bibitem[\protect\citeauthoryear{Swanson and Hern\'{a}n}{2014}]{Swanson:2014}
\begin{barticle}[author]
\bauthor{\bsnm{Swanson},~\bfnm{Sonja~A.}\binits{S.~A.}} \AND
  \bauthor{\bsnm{Hern\'{a}n},~\bfnm{Miguel~A.}\binits{M.~A.}}
(\byear{2014}).
\btitle{Think Globally, Act Globally: An Epidemiologist's Perspective on
  Instrumental Variable Estimation}.
\bjournal{Statistical Science}
\bvolume{29}
\bpages{371-374}.
\end{barticle}
\endbibitem

\bibitem[\protect\citeauthoryear{Tamer}{2010}]{tamer2010partial}
\begin{barticle}[author]
\bauthor{\bsnm{Tamer},~\bfnm{Elie}\binits{E.}}
(\byear{2010}).
\btitle{Partial identification in econometrics}.
\bjournal{Annu. Rev. Econ.}
\bvolume{2}
\bpages{167--195}.
\end{barticle}
\endbibitem

\bibitem[\protect\citeauthoryear{Tchetgen~Tchetgen and
  VanderWeele}{2012}]{TchetgenTchetgenVanderWeele2010}
\begin{barticle}[author]
\bauthor{\bsnm{Tchetgen~Tchetgen},~\bfnm{Eric~J.}\binits{E.~J.}} \AND
  \bauthor{\bsnm{VanderWeele},~\bfnm{Tyler~J}\binits{T.~J.}}
(\byear{2012}).
\btitle{On causal inference in the presence of interference}.
\bjournal{Statistical Methods in Medical Research}
\bvolume{21}
\bpages{55--75}.
\end{barticle}
\endbibitem

\bibitem[\protect\citeauthoryear{Vanderweele}{2008}]{VanderWeele2008-StatMed}
\begin{barticle}[author]
\bauthor{\bsnm{Vanderweele},~\bfnm{TJ}\binits{T.}}
(\byear{2008}).
\btitle{Ignorability and stability assumptions in neighborhood effects
  research.}
\bjournal{Statistics in Medicine}
\bvolume{27}
\bpages{1934--1943}.
\end{barticle}
\endbibitem

\bibitem[\protect\citeauthoryear{VanderWeele et~al.}{2013}]{Vanderweele:2013}
\begin{barticle}[author]
\bauthor{\bsnm{VanderWeele},~\bfnm{Tyler~J.}\binits{T.~J.}},
  \bauthor{\bsnm{Hong},~\bfnm{Guanglei}\binits{G.}},
  \bauthor{\bsnm{Jones},~\bfnm{Stephanie~M.}\binits{S.~M.}} \AND
  \bauthor{\bsnm{Brown},~\bfnm{Joshua~L.}\binits{J.~L.}}
(\byear{2013}).
\btitle{Mediation and Spillover Effects in Group Randomized Trials: A Case
  Study of the 4Rs Educational Intervention}.
\bjournal{Journal of the American Statistical Association}
\bvolume{108}
\bpages{469--481}.
\end{barticle}
\endbibitem

\end{thebibliography}

\end{document}